\documentclass[11pt]{article}
\usepackage{epsfig,amsfonts,amsthm}
\usepackage{amsmath,amssymb}
\usepackage{float}
\usepackage{cite}
\usepackage{xcolor}
\usepackage{braket}
\usepackage{booktabs}
\usepackage{makecell}
\usepackage{multirow}
\usepackage{makecell}
\usepackage{mdframed}
\newtheorem{theorem}{Theorem}

\newcommand{\be}{\begin{equation}}
\newcommand{\ee}{\end{equation}}
\newcommand{\bea}{\begin{eqnarray}}
\newcommand{\eea}{\end{eqnarray}}

\newcommand{\lr}[1]{ \langle #1 \rangle}

\newcommand{\Z}{\mathbb{Z}}

\providecommand{\id}{{\boldsymbol{1}}}
\newcommand{\diag}{\mathrm{diag}}
\newcommand{\Aut}{\mathrm{Aut}}
\newcommand{\fdf}[2]{(\phi^\dagger_{#1}\phi_{#2})}
\newcommand{\fdfn}[2]{\phi^\dagger_{#1}\phi_{#2}}

\definecolor{darkred}{rgb}{0.7,0.0,0.0}

\newcommand{\gray}{\color{gray}}
\definecolor{darkgreen}{rgb}{0.0,0.5,0.0}

\def\lsim{\mathrel{\rlap{\lower4pt\hbox{\hskip1pt$\sim$}}
\raise1pt\hbox{$<$}}}         
\def\gsim{\mathrel{\rlap{\lower4pt\hbox{\hskip1pt$\sim$}}
\raise1pt\hbox{$>$}}}         

\title{Symmetries for the 4HDM. II. Extensions by rephasing groups}

\author{Jiazhen Shao$^1$\thanks{E-mail: shaojzh5@gmail.com}\, , 
Igor P. Ivanov$^1$\thanks{E-mail: ivanov@mail.sysu.edu.cn}\,, and\,
Mikko Korhonen$^2$
\\
{\small $^1$ School of Physics and Astronomy, Sun Yat-sen University, 519082 Zhuhai, Guangdong, P.R. China}
\\
{\small $^2$ Shenzhen International Center for Mathematics, Southern University of Science and Technology, } \\
{\small Shenzhen 518055, Guangdong, P.R. China}
}

\begin{document}
\maketitle

\bigskip
\begin{abstract}
We continue classification of finite groups which can be used as symmetry group 
of the scalar sector of the four-Higgs-doublet model (4HDM). 
Our objective is to systematically construct non-abelian groups via the group extension procedure,
starting from the abelian groups $A$ and their automorphism groups $\Aut(A)$.
Previously, we considered all cyclic groups $A$ available for the 4HDM scalar sector.
Here, we further develop the method and apply it to extensions 
by the remaining rephasing groups $A$, namely $A = \Z_2\times\Z_2$, $\Z_4\times \Z_2$, and $\Z_2\times \Z_2\times \Z_2$.
As $\Aut(A)$ grows, the procedure becomes more laborious, but we prove an isomorphism theorem which 
helps classify all the options.
We also comment on what remains to be done to complete the classification of all finite 
non-abelian groups realizable in the 4HDM scalar sector without accidental continuous symmetries.
\end{abstract}

\newpage
\tableofcontents
\bigskip

\hrule

\section{Introduction}

The $N$-Higgs-doublet model (NHDM) is a popular framework
for building New Physics models beyond the Standard Model (SM).
Based on the simple idea of successive ``Higgs generations'' shaped 
by global symmetries, NHDMs can address many shortcomings of the SM
and offer remarkably rich phenomenological and astroparticle signatures.

The most famous examples are the 2HDM proposed by T.D.~Lee in 1973 \cite{Lee:1973iz}
and routinely used nowadays in the LHC searches \cite{Branco:2011iw},
and the 3HDM suggested first by S.~Weinberg in 1976 \cite{Weinberg:1976hu}
and studied since then in hundreds of papers, see the recent review \cite{Ivanov:2017dad}.
The literature on models with $N > 3$ Higgs doublets is less known to the community but nonetheless impressive.
In particular, the four-Higgs-doublet model (4HDM) was first considered in the seminal 1977 paper 
by Bjorken and Weinberg \cite{Bjorken:1977vt}, and nearly hundred papers explored various 4HDM-based
models since then, see our previous paper \cite{Shao:2023oxt} for a 4HDM literature overview.

Why go beyond the 2HDM?
From the model-building perspective, a powerful novelty of NHDMs beyond two Higgs doublets
is the rich list of symmetry-based options one has at one's disposal.
By ``symmetry-based options'' we mean symmetry groups themselves
such as permutation groups $S_n$ and $A_n$ \cite{Derman:1979nf}
and $\Delta(27)$ \cite{Segre:1978ji} first discussed as early as the late 1970s, 
specific representation choices for scalars and fermions, see for example a systematic study of options 
in the $A_4$ and $S_4$ 3HDM \cite{GonzalezFelipe:2013xok}, 
the residual symmetry group at the minimum of the Higgs potential \cite{Ivanov:2014doa}
and its consequences for the fermion sector \cite{Leurer:1992wg,GonzalezFelipe:2013xok,GonzalezFelipe:2014mcf},
new forms of $CP$ symmetry of higher order \cite{Ivanov:2015mwl,Haber:2018iwr}, 
and additional opportunities for $CP$ violation \cite{Branco:1999fs}
such as geometric $CP$ violation in the scalar sector \cite{Branco:1983tn,deMedeirosVarzielas:2011zw,deMedeirosVarzielas:2012rxf,Ivanov:2013nla}. 
All these features lead to intriguing observable signatures.
An extensive discussion of the symmetry options available in the 3HDM and their phenomenological consequences
can be found, for example, in \cite{Ivanov:2017dad,Das:2014fea,Darvishi:2021txa,Das:2021oik}.
In fact, almost all 3HDM and 4HDM examples studied in literature
featured phenomenological and astroparticle signals induced by symmetries.

In short, symmetries are the treasure trove of the NHDMs.
It is therefore important to know which symmetry groups are available for a given class of models
and how these symmetries are broken upon minimization.
Within the 3HDM, many of these questions have already been answered.
We highlight in particular the classification of all finite symmetry groups realizable in the scalar sector 
of the 3HDM \cite{Ivanov:2011ae,Ivanov:2012fp,Ivanov:2012ry} and their symmetry-breaking pattern 
for each finite group \cite{Ivanov:2014doa}. 
These are complemented by insights into continuous symmetry groups, including accidental symmetries 
of the scalar potential, worked out in \cite{Darvishi:2019dbh,Darvishi:2021txa}.
All in all, the symmetry content of the 3HDMs is now well understood, and the interesting question --- 
or rather a promising research program --- is to track all phenomenological and cosmological consequences
of each symmetry option.
In contrast, the symmetry options in the 4HDM remained unexplored, with most papers limited to isolated examples.

In our recent work \cite{Shao:2023oxt}, we started a systematic classification 
of finite non-abelian symmetry groups for the 4HDM scalar sector. 
Not satisfied with the trial and error approach, we decided to apply the method which worked so successfully
for the 3HDM and began a systematic exploration of non-abelian groups obtained from abelian groups 
via the group-theoretic procedure called group extension.
In \cite{Shao:2023oxt}, we described the method and applied it to all cyclic groups available in the 4HDM scalar sector.
In the present work we extend the same method to constructing non-abelian extensions 
of the remaining rephasing groups, which are products of cyclic groups.

The structure of this paper is the following. In the next Section, we describe the general
procedure of constructing non-abelian groups using group extensions and apply it to the 4HDM.
We also prove an important Theorem which will later help us significantly reduce the number of cases to consider.
The short Section~\ref{section-Z2Z2} deals with extensions by $\Z_2\times\Z_2$, which borrows a lot from the 3HDM case.
Section~\ref{section-Z4Z2} discussed in detail extensions by $\Z_4\times\Z_2$, with its numerous non-trivial features.
In Section~\ref{section-Z2Z2Z2}, we build extensions based on $\Z_2\times\Z_2\times\Z_2$.
We present the master table of our results in Section~\ref{section-conclusions} 
and wrap up the discussion with an outlook of future work.

\section{Symmetries in the 4HDM scalar sector: general features}

\subsection{The need of constructive methods}

How do we build a model based on a specific finite group $G$ of global symmetries?
The standard approach is to assign the fields --- in our case, the Higgs doublets --- to 
a specific representation of $G$ and to construct all interaction terms of the lagrangian 
which are invariant under $G$.

This procedure, despite looking straightforward, has important pitfalls.
It may happen that, by imposing $G$, we end up with a lagrangian invariant under additional
symmetries, so that the true symmetry content of the model is larger than $G$.
A notorious situation is when we impose a finite group $G$ 
but the scalar potential possesses an accidental continuous symmetry, 
which can be spontaneously broken upon minimization leading to unwanted Goldstone bosons. 
Examples of this situation are known within the 2HDM and 3HDM.
For instance, the nearest-neighbor-interaction model constructed in \cite{Branco:2010tx}
and used in later papers was based on the $\Z_4$ symmetry imposed on the 2HDM scalar and Yukawa sector. 
However, as demonstrated in \cite{Ivanov:2013bka}, 
imposing $\Z_4$ always leads to a continuous symmetry.
Another well-known example is that imposing an single $CP$ symmetry of order higher than 4 
on the 2HDM \cite{Ferreira:2010hy,Ferreira:2010yh} and 3HDM \cite{Ivanov:2018qni} scalar potentials 
leads to a continuous Higgs family symmetry.
A systematic analysis of the accidental continuous symmetries in the 3HDM scalar sector can be found in \cite{Darvishi:2019dbh}.
Thus, when we classify finite symmetry groups available within a given class of models, 
we must assure that each $G$ represents the full symmetry content of the model and 
that the resulting potential does not contain additional accidental symmetries.
Following the notation of \cite{Ivanov:2011ae,Ivanov:2012fp}, we call such groups realizable.

Certain symmetry groups and the potentials invariant under them can be constructed in a straightforward way.
Such educated guesses are fully justified for an exploratory study. 
However, when attempting to cover all possible cases symmetry-based phenomenological situations within a given class of models,
one should ask for the full list of realizable symmetry groups available.
Such a list can only be obtained with a systematic constructive procedure.

To give an example, the 3HDMs based on global symmetry groups were first explored in late 70's and early 80's,
see \cite{Ivanov:2017dad} for a brief historical overview. Many choices for non-abelian groups $G$ were used, 
but it was unclear when to stop the search. 
The answer came only in 2012 when the papers \cite{Ivanov:2011ae,Ivanov:2012ry,Ivanov:2012fp}
presented the full classification of finite Higgs family symmetry groups and $CP$ symmetries.
In particular, these papers identified two symmetry-based 3HDMs which had not been recognized in earlier papers:
$\Sigma(36)$ 3HDM, see e.g. \cite{deMedeirosVarzielas:2021zqs} and CP4-symmetric 3HDM \cite{Ivanov:2015mwl}.

With four Higgs doublets, there is now a long record of 4HDM publications based on educated guesses for the group $G$, 
see the literature overview in \cite{Shao:2023oxt}. 
However no systematic classification of all available options exist so far.
Interesting 4HDMs with potentially peculiar phenomenology may still be hiding,
and only a systematic constructive procedure can reveal all these examples.

\subsection{Building non-abelian groups via group extension}\label{subsection-extension}

Let us now briefly review the group extension method used in \cite{Ivanov:2012ry,Ivanov:2012fp} 
to derive the complete list of discrete symmetry groups in the 3HDM scalar sector.

Any finite non-abelian group $G$ contains abelian subgroups $A$. 
By identifying possible $A$'s, one can already learn much about possible $G$.
In \cite{Ivanov:2011ae}, an algebraic technique based on the so-called Smith normal form was developed, 
which is capable of identifying the rephasing symmetry group of any lagrangian, not limited to NHDMs.
For the 3HDM, it yielded the following list of realizable finite abelian groups which are subgroups\footnote{The need and subtleties of working inside $PSU(N)$ instead of $SU(N)$ for the NHDM scalar sector are discussed in \cite{Ivanov:2011ae} and, in a way adapted to the 4HDM, in our previous paper \cite{Shao:2023oxt}.} of $PSU(3)$:
\begin{equation}
\mbox{$A$ in 3HDM:}\qquad \Z_2\,,\quad 
\Z_3\,,\quad 
\Z_4\,,\quad 
\Z_2\times \Z_2\,,\quad 
\Z_3\times \Z_3\,.
\label{3HDM-abelian}
\end{equation}
Since the orders of all these abelian groups contain only primes, 2 and 3,
the order of any non-abelian finite group is $|G| = 2^a 3^b$.
According to Burnside’s $p^aq^b$-theorem, such group $G$ is solvable; 
for a physicist-friendly introduction to solvable groups, see section~3 of Ref.~\cite{Ivanov:2012fp}. 
Together with additional group-theoretic arguments, Ref.~\cite{Ivanov:2012fp} established
that, within the 3HDM scalar sector, any finite non-abelian group $G$ must contain a normal maximal 
abelian subgroup. As a result, the factor group $G/A \subseteq \Aut(A)$, the automorphism group of $A$.
Thus, one arrives at a systematic procedure to classify all realizable groups $G$ in the 3HDM: 
\begin{itemize}
\item Take $A$ from the list Eq.~\eqref{3HDM-abelian}, compute $\Aut(A)$, list all its subgroups $K \subseteq \Aut(A)$.
\item
If $G/A \simeq K$, then $G$ can be constructed as an extension\footnote{Incidentally, in our previous paper \cite{Shao:2023oxt} we called the same construction ``extension of $A$ by $K$'' instead of ``extension of $K$ by $A$''. This description was less accurate. In fact, we take $K$ and ``blow up'' each element $k \in K$ to a coset $kA$. In this way, we extend $K$ by $A$, which is the standard formulation in mathematical texts. This is why the title of the present paper contains ``extensions by rephasing groups'', not ``extensions of rephasing groups''.} of $K$ by $A$ denoted as $A\,.\,K$.
Even for a specific $A$ and a specific $K$, the extension procedure is not unique. One needs to construct all the cases explicitly by defining how the generators of $K$ act on the generators of $A$. The most well-known type of non-abelian extensions are semidirect products $A \rtimes K$, also called split extensions. There also can exist non-split extensions;
for a pedagogical exposition and illustrative examples, see section~3 of \cite{Shao:2023oxt}.
\item
By checking all $A$'s, all $K$'s, performing all possible extensions, and verifying that 
the resulting potential does not acquire an accidental symmetry, one obtains the full list 
of finite non-abelian groups $G$ for the 3HDM scalar sector.
This list derived in \cite{Ivanov:2012ry,Ivanov:2012fp} is not long:
$S_3$, $D_4$, $A_4$, $S_4$, $\Delta(54)/\Z_3$, and $\Sigma(36)$.
\end{itemize}
Let us now discuss whether this procedure can be applied to the classification of non-abelian symmetry groups in the 4HDM.
One begins with the list of finite abelian subgroups of $PSU(4)$ realizable within the 4HDM.
We remind the reader that $PSU(4)$ is defined by taking all unitary matrices acting in $\mathbb{C}^4$
and identifying all transformations which differ only by the overall phase factor.
Group-theoretically, we factor $U(4)$ by its center, $U(1)$, and obtain as $U(4)/U(1) \simeq PSU(4)$.
Note also that it is not $SU(4)$, for $SU(4)$ contains matrices of the form $i^k\,\cdot \id_4$, with $k = 0,1,2,3$,
which form the group $\Z_4$, the center of $SU(4)$. Thus we need to factor $SU(4)$ by its center 
to arrive at $SU(4)/Z(SU(4)) \simeq PSU(4)$.

With this preliminary remark, we see that there are two sorts of abelian groups $A \in PSU(4)$, see more details in \cite{Ivanov:2011ae}. 
First, we have the groups of transformations which, in a suitable basis, can be represented by rephasing of individual doublets. 
Clearly, such groups are abelian both in $SU(4)$ and $PSU(4)$.
The full list of rephasing groups realizable in the 4HDM scalar sector was already established in \cite{Ivanov:2011ae}.
Focusing on finite rephasing symmetry groups, we get
\begin{equation}
\mbox{rephasing $A$ in 4HDM:}\qquad \Z_k\ \mbox{with}\ k = 2,\dots,8\,,\quad 
\Z_2\times \Z_2\,,\quad 
\Z_4\times \Z_2\,,\quad 
\Z_2\times \Z_2 \times \Z_2\,,
\label{4HDM-abelian-1}
\end{equation}
or, in simple words, all abelian groups of order at most 8.

Groups $A$ of the second family are abelian in $PSU(4)$ but their full preimage $\hat{A} \in SU(4)$ 
is a non-abelian group of special type: a finite nilpotent group of class 2, that is, its commutator group
$[\hat{A},\hat{A}]$ is non-trivial but lies inside its center $Z(\hat{A})$, which coincides with $Z(SU(4))$.
By passing from $SU(4)$ to $PSU(4)$, the center is factored out, and this is why $A = \hat{A}/\Z_4$ is abelian in $PSU(4)$.

Within the 3HDM, we had only one option for such a group: $\Z_3\times \Z_3$, the last group in the list of Eq.~\eqref{3HDM-abelian}.
Within the 4HDM, we have discovered three\footnote{In the previous paper \cite{Shao:2023oxt}, we identified only the first 
	of these three options. It was only during the work on the present paper, that we discovered two more cases 
	of realizable finite abelian group listed in Eq.~\eqref{4HDM-abelian-2}. 
	There are strong arguments suggesting that these three groups exhaust the additional finite abelian symmetry groups of $PSU(4)$ realizable 
	in the 4HDM. Indeed, theorem 1 of Ref.~\cite{Ivanov:2011ae} asserts that the additional groups $A$ are such that
	$\hat{A}$ is a finite nilpotent group of class 2 (i.e. the commutator $[\hat{A},\hat{A}]$ lies in the center $Z(A)$),
	that the exponent of $\hat{A}$ is divisible by $N=4$ and divides $N^2=16$. Moreover, $Z(\hat{A}) = Z(SU(4))$.
	These conditions are quite restrictive to make it plausible that no additional $A$'s with non-abelian $\hat{A}$'s exist.
	We do not have a complete proof, though, and delegate to a future work both the the task of establishing the proof and
	explicit construction of extensions by each of these additional groups. In the present paper, 
	this list of these additional groups should be considered as provisional.} 
	such options so far:
\begin{equation}
	\mbox{extra $A$ in 4HDM:}\qquad \Z_4\times \Z_4\,,\quad 
	\Z_4\times \Z_2 \times \Z_2\,,\quad (\Z_2)^4\,. \label{4HDM-abelian-2}
\end{equation}
We stress that the transformations of these groups cannot be reduced to pure rephasing factors.
For example, consider the first one, $\Z_4\times \Z_4$.
Its full pre-image in $SU(4)$ is the non-abelian group $(\Z_4 \times \Z_4)\rtimes \Z_4$ of order 64
generated by
\begin{equation}
a = \sqrt{i}\left(\begin{array}{cccc} 
	1 & 0 & 0 & 0\\ 
	0 & i & 0 & 0\\ 
	0 & 0 & -1 & 0\\ 
	0 & 0 & 0 & -i\\ 
\end{array}\right)\,,\quad
b = \sqrt{i}\left(\begin{array}{cccc} 
	0 & 1 & 0 & 0\\ 
	0 & 0 & 1 & 0\\ 
	0 & 0 & 0 & 1\\ 
	1 & 0 & 0 & 0\\ 
\end{array}\right)\,.\label{generators-Z4Z4}
\end{equation}
The factor $\sqrt{i}$ is introduced in $a$ and $b$ to make sure that $\det a = \det b = 1$.
The two generators do not commute inside $SU(4)$ but they do in $PSU(4)$, 
since their commutator $[a,b]\equiv aba^{-1}b^{-1} = -i\cdot \id_4$ belongs to the center of $SU(4)$.
Therefore, factoring this group by the center of $SU(4)$ produces the abelian group $\Z_4\times \Z_4$. 
Similar situation occurs when we are dealing with subgroups $\Z_4\times\Z_2\times\Z_2$ and $(\Z_2)^4$ in $PSU(4)$ whose pre-images in $SU(4)$ are not abelian. 

Having identified finite abelian groups for the 4HDM scalar sector, we can try to apply the 
above group extension procedure. 
Any finite non-abelian group $G$ can have abelian subgroups only from the lists in Eqs.~\eqref{4HDM-abelian-1} and \eqref{4HDM-abelian-2}. 
Unlike in the 3HDM case, we now encounter four different prime factors: 2, 3, 5, and 7. 
Thus, Burnside's $p^aq^b$-theorem is no longer applicable,
and we cannot guarantee that any group $G$ contains a normal maximal abelian subgroup.
For example, we cannot exclude the group $A_5$ of order 120, which is simple and therefore does not possess any normal subgroup.
We are forced to conclude that, by following the same arguments as for the 3HDM, 
we will not be able to list {\em all} realizable non-abelian groups of the 4HDM. 

However, we can set a less ambitious goal: classify all realizable finite non-abelian groups in the 4HDM
which are constructible via this extension procedure. It is this task that we started in the previous paper 
\cite{Shao:2023oxt} and continue in the present work.

\begin{table}[H]
\centering
\begin{tabular}[t]{cc}
	\toprule
	$A$ & $\Aut(A)$ \\
	\midrule
	$\Z_2$ & $\{e\}$ \\
	$\Z_3$ & $\Z_2$ \\
	$\Z_4$ & $\Z_2$ \\
	$\Z_5$ & $\Z_4$ \\
	$\Z_6$ & $\Z_2$ \\
	$\Z_7$ & $\Z_6$ \\
	$\Z_8$ & $\Z_2\times \Z_2$ \\
	\bottomrule
\end{tabular}
\qquad
\begin{tabular}[t]{cc}
	\toprule
	$A$ & $\Aut(A)$ \\
	\midrule
	$\Z_2 \times \Z_2$ & $S_3$ \\
	$\Z_2 \times \Z_4$ & $D_4$ \\
	$\Z_2 \times \Z_2 \times \Z_2$ & $GL(3,2)$ \\[1mm]
	\midrule\\[-3mm]
	$\Z_4 \times \Z_4$ & $GL(2,\Z_4)$ \\
	$\Z_4 \times \Z_2 \times \Z_2$ & ${\tt SmallGroup(192, 1493)}$ \\
	$(\Z_2)^4$ & $A_8$ \\
	\bottomrule
\end{tabular}
\caption{The list of finite abelian symmetry groups $A$ of the 4HDM scalar sector
	and their automorphism groups $\Aut(A)$. In this paper, we construct extensions by the first three lines of the right part of the Table.
The list in the bottom half of the right part is provisional.}
\label{table-abelian}
\end{table}
Having this task in mind, we proceed by listing in Table~\ref{table-abelian} 
all abelian groups $A$ from Eqs.~\eqref{4HDM-abelian-1} and \eqref{4HDM-abelian-2}, together with their automorphism groups $\Aut(A)$. 
The left part of this Table contains cyclic groups $A$ from $\Z_2$ to $\Z_8$. 
Their automorphism groups are abelian and the extensions are relatively easy to build in a straightforward way;
in \cite{Shao:2023oxt} we reported the results.

The right part of the Table contains products of cyclic groups. 
The first three are rephasing groups and, moreover, they are all such groups in the 4HDM scalar sector. 
Extensions by them will be the subject of the present paper.
The last three are obtained from nilpotent groups of class 2 in $SU(4)$.
As already indicated, we do not yet have the proof that we have found all such groups in the 4HDM,
so this list is provisional. In any case, the proof and 
explicit constructions of the extensions based on them is postponed to a follow-up paper.

Even focusing on the first three groups of the right half of Table~\ref{table-abelian},
we remark that construction of their extensions is more challenging than for cyclic groups.
First, the automorphism groups of products of cyclic groups are non-abelian, which naturally are more difficult to study. 
Second, these groups can easily become very large.
For example, $\Aut((\Z_2)^3)\simeq GL(3,2) \simeq SL(3,2) \simeq PSL(2,7)$ has the order $2^3\times 3\times 7 = 168$.
Third, if $K_1 \subset \Aut(A)$ and $K_2 \subset \Aut(A)$ are two non-isomorphic subgroups,
then in general the extensions $G_1 = A\,.\,K_1$ and $G_2 = A\,.\,K_2$ are also non-isomorphic
and can lead to physically distinct multi-Higgs models. Thus, we need to explore all subgroups $K$ 
of the automorphism group $\Aut(A)$ and, for each $K$, study all the extensions $A\,.\,K$. 
For large $\Aut(A)$, the subgroup structure can be complicated, leading to a huge amount of work listing and studying all $K$'s in $\Aut(A)$. For example, 
using the computer algebra system {\ttfamily GAP}~\cite{GAP}, we found that $\Aut((\Z_2)^3)\simeq GL(3,2)$ has 179 subgroups.
Checking all of them one by one and constructing all possible extensions would be impractical.

In this situation, two observations come to rescue and, eventually, make the analysis manageable.
First, even if $\Aut(A)$ is large, we can start by picking up a cyclic subgroup $K \subset \Aut(A)$ generated by a certain $k \in \Aut(A)$,
in particular, by one of its generators, and try extending it by $A$. 
It turns out that some choices immediately lead to continuous accidental symmetries
and can be safely eliminated.

Second, although $\Aut(A)$ can contain a large number of subgroups $K$, many of them lead to extensions
which result in the same 4HDM potentials, up to a basis change. 
In the next subsection, we prove a theorem which clarifies this statement and helps us further reduce 
the number of cases to consider.

\subsection{Extensions of conjugate groups are isomorphic}\label{subsection-theorem}

In this paper, our main interest is in finite groups $G$ with an abelian normal subgroup $A$ which is self-centralizing.
The latter property means that the centralizer\footnote{For the definition and the role of the centralizer, as well as the importance and consequences of having 
a normal self-centralizing abelian subgroups, we refer to section 3 of Ref.~\cite{Ivanov:2012fp},
which offers a physicist-friendly introduction to the basics of the theory of solvable groups.} of $A$ in $G$ coincides with $A$ itself: 
$C_G(A) = A$.
For such a group, the natural homomorphism $\pi: G \rightarrow \Aut(A)$ has kernel equal to $A$. 
We denote then $G = A\,.\,K$, where $K = \pi(G) \cong G/A$ is the action of $G$ on $A$.

The following result shows that it will suffice to consider the groups $K$ up to conjugacy in $\Aut(A)$.

\begin{theorem}
Let $A$ be a finite abelian group and $K_1$, $K_2$ be two subgroups of $\Aut(A)$ which are conjugate to each other, meaning there exists $q \in \Aut(A)$ such that $q^{-1}K_1 q = K_2$. 

Suppose $G_1 = A\,.\,K_1$ is a group extension such that $C_{G_1}(A) = A$, and $K_1$ is the action of $G_1$ on $A$. Then $G_1$ is isomorphic to a group of the form $G_2 = A\,.\,K_2$, where $C_{G_2}(A) = A$, and $K_2$ is the action of $G_2$ on $A$.
\end{theorem}

\begin{proof}
We consider $A$ as an additive group $(A,+)$. For $K \leq \Aut(A)$, a \emph{$2$-cocycle} for $K$ is a map $f: K \times K \rightarrow A$ satisfying $f(x,1) = 0 = f(1,x)$ and $$f(xy,z) + f(x,y) = x(f(y,z)) + f(x,yz)$$ for all $x,y,z \in K$. Given a $2$-cocycle $f$, we can form the group $G_{f,K} = \{(a,x) : a \in A, x \in K \}$ with the group operation defined by $$(a,x) \cdot (a',x') = (a+x(a')+f(x,x'),xx')$$ for all $a,a' \in A$ and $x,x' \in K$. For the fact that $G_{f,K}$ is a group, see \cite[Theorem 7.30]{Rotman}. 

We can identify $A$ as a normal subgroup of $G_{f,K}$ via the isomorphism $a \mapsto (a,1)$. We have $$(a_0,x) \cdot (a,1) \cdot (a_0,x)^{-1} = (x(a),1)$$ for all $a,a_0 \in A$ and $x \in K$, therefore the action of $G_{f,K}$ on $A$ is equal to the group $K$. This also implies that $C_{G_{f,K}}(A) = A$. It follows from \cite[Theorem 7.30]{Rotman} that if $G = A\,.\,K$ is a group extension such that $C_G(A) = A$ and $K$ is the action of $G$ on $A$, then $G$ is isomorphic to $G_{f,K}$ for some $2$-cocycle $f$. 

Therefore, for the proof of the theorem we may assume that $G_1 = G_{f,K_1}$ for some $f$. Since $q^{-1}K_1 q = K_2$, we can define $f': K_2 \times K_2 \rightarrow A$ by $$f'(x,y) = q^{-1} \left(f(qxq^{-1}, qyq^{-1})\right)$$ for all $x,y \in K_2$. Using the fact that $f$ is a $2$-cocycle, a straightforward check shows that $f'$ is a $2$-cocycle for $K_2$. Now define a map $\tau: G_{f,K_1} \rightarrow G_{f',K_2}$ by $$\tau(a,x) = (q^{-1}(a), q^{-1} x q)$$ for all $a \in A$ and $x \in K_1$. It follows from the definitions that $\tau$ is an isomorphism, so by taking $G_2 = G_{f',K_2}$ we have $G_1 \cong G_2$, where $G_2$ has the required properties.\end{proof}

To make this formal proof more accessible to the physics community, we would like to point out that in the case of split extensions, also called the semidirect products $G = A\rtimes K$, the proof simplifies. In this case, the group $G$ contains 
not only the normal subgroup $A$ but also a subgroup $H$ isomorphic to $K$ which is a complement for $A$ in $G$, that is, 
$H \cap A = \{e\}$ and $G = AH$. Then one can define in a straightforward way the group operation 
not only on the set of $A$-cosets, but also on set of representative elements of these cosets. 
A simple way to obtain the split extension is to take
the trivial 2-cocycle $f(x,x') = 0$ for all $x,x' \in K$. 
In this case, the multiplication law is fully defined by the action
$x(a)$, that is, how $x \in K$ permutes the elements of $A$. The theorem then reduces to the calculation which relates the elements $g_1 \in G_1$
to the elements $g_2 \in G_2$. 

However when we build a non-split extension, the set of elements of $G$ of the form $(0,x)$ is not closed under the same group operation as $x \in K$. The 2-cocycle $f(x,x')$ is the construction that describes this failure to reproduce the group structure of $K$ inside $G$. Thus, to define the structure of the group $G$, it is not enough to specify $x(a)$, the action of elements of $K$ on $A$; we also need to define the 2-cocycle $f(x,x')$. As a result, in order to prove the isomorphism between $G_1$ and $G_2$, we need to demonstrate not only the relation between the actions $x_1(a)$, $x_1 \in K_1$, and $x_2(a)$, $x_2 \in K_2$, but also the unambiguous link between the corresponding 2-cocycles $f$ and $f'$. This is what the body of the proof does.

This theorem allows us to reduce the number of extensions we need to consider, especially when $\Aut(A)$ is large.
Namely, what we need to list is not all individual subgroups of $\Aut(A)$ but only 
all conjugacy classes of these subgroups. 
For each conjugacy class, we can select one representative subgroup and find all of its extensions of $A$.
Since all subgroups from a conjugacy class are conjugate to each other, it follows from the theorem just proved that the list of extensions of any other subgroup from the same conjugacy class will be the same. 

There is, however, an important caveat which does not render the classification problem as easy as it may sound.
When we apply the above result to construction of the symmetry-based 4HDMs, we deal not with abstract groups
but with their four-dimensional representations, that is, subgroups of $PSU(4)$.
Even if $A$, $K_1$, and $K_2$ can be faithfully represented as groups of transformations from $PSU(4)$,
the transformation $q$ linking $K_1$ and $K_2$ and, consequently, the transformation $\tau$ linking $G_1$ and $G_2$
are not guaranteed to belong to $PSU(4)$. 
In our previous paper \cite{Shao:2023oxt}, we encountered examples of group transformations 
which do not fit the desired representation.
For example, the group $A = \Z_8 \in PSU(4)$ has the automorphism group $\Aut(\Z_8) = \Z_2\times \Z_2$.
We can take any non-trivial element $q$ from $\Aut(\Z_8)$ and define its action on the group $\Z_8$
in abstract group-theoretic terms. However, once we write $\Z_8$ as a subgroup of $PSU(4)$ and try to construct
$q \in PSU(4)$ satisfying the desired relations, we end up with a system of equations which does not have solutions. 
Thus, such $q$ does not fit $PSU(4)$.

If it happens that $\tau$, which maps $G_1 \to G_2$ by conjugation, can indeed be represented by a $PSU(4)$ transformation, 
then the situation simplifies.
Indeed, the invariance of $V_1(\phi)$ under $G_1$ means
that $V_1(g_1(\phi)) = V_1(\phi)$ for any $g_1 \in G_1$. 
The transformation $\tau$ now acts in the same space and defines a basis change: $\phi \mapsto \tau(\phi)$.
This is not a symmetry of the potential: $V_1(\tau(\phi)) \equiv V_2(\phi) \not = V_1(\phi)$.
However the potential $V_2(\phi)$ defined in this way is invariant under the group $G_2$.
Indeed, picking up $g_2 \in G_2$ and representing it as $\tau^{-1}g_1 \tau$ for some $g_1 \in G_1$, we obtain
\begin{equation}
	V_2\left(g_2(\phi)\right) = V_2\left(\tau^{-1}(g_1(\tau(\phi)))\right) = V_1\left(g_1(\tau(\phi))\right) =
	V_1\left(\tau(\phi)\right) = V_2(\phi)\,.
\end{equation}
Therefore, the potential $V_2(\phi)$ invariant under $G_2$ has the same symmetry content 
(that is, their symmetry groups are related by a basis change) as $V_1(\phi)$ invariant under $G_1$.
Since the two potentials, $V_1(\phi)$ and $V_2(\phi) =  V_1(\tau(\phi))$, are related by a mere basis change,
their physics consequences are identical. 
On the contrary, if it happens that $\tau$ linking $G_1$ and $G_2$ cannot be represented by an $PSU(4)$ transformation,
than the potentials $V_1$ and $V_2$ have different physical consequences, even if their symmetry groups
carry the same labels.

With this helpful result, we update our procedure for building 4HDM models based on extensions of the type $A\,.\,K$,
where $K \subseteq \Aut(A)$. We need to consider not the individual subgroups $K$ nor the entire conjugacy classes 
of $K$ inside $\Aut(A)$, but the conjugacy classes in which we only use transformations $\tau \in \Aut(A)$ 
expressible as $PSU(4)$ transformations. It is then sufficient to consider only one representative $K$ 
from each such conjugacy class and build all the group extensions available.

\subsection{The strategy of the work}

In this paper, we will apply the group extension procedure outlined in Section~\ref{subsection-extension} 
to the first three abelian groups in the right part of Table~\ref{table-abelian}.
The sequence of steps will be similar to our previous work \cite{Shao:2023oxt}.
\begin{itemize}
\item 
Pick up the abelian group $A \subset PSU(4)$ and write down its generators. 
It is convenient to represent them as $SU(4)$ rather than $PSU(4)$ transformations, 
but we need to keep in mind that all relations are defined modulo to
the center of $SU(4)$, which is the group $\Z_4$ generated by $i\cdot \id_4$.
\item 
It may happen that there are more than one nonequivalent ways a given abelian group $A$ can be implemented in the 4HDM scalar sector
without leading to accidental continuous symmetries. For example, in our previous work \cite{Shao:2023oxt} we found
that $\Z_4$ can be implemented in three distinct ways, which cannot be linked by any basis change. 
Therefore, these are three distinct $\Z_4$-invariant 4HDM models.
Each implementation leads to the Higgs potential with different number of free parameters and, eventually, different
options for the non-abelian extensions.
In order to be sure that we do not miss any implementation, we rely on our own Python code {\tt 4HDM Toolbox} \cite{TheCode},
which was described in \cite{Shao:2023oxt} and is freely available at GitHub.
\item
List all the conjugacy classes of subgroups of $\Aut(A)$ using only such transformations which can be expressed as $PSU(4)$ transformations. 
For each class, take a representative subgroup with its generators,
which we generically write as $b$.
By defining how each generator $b$ acts on $A$ and by choosing whether the properties of $b$'s, now seen as the elements of the extension group,
reproduce the properties of $b$'s inside the parent group $\Aut(A)$, construct all possible non-abelian extensions.
\item 
Using a specific implementation of the group $A$, which is defined by the expression of its generators $a_i$,
write the action of $b$ on $a_i$ in the form of matrix equations. Solve these equations.
If a unitary solution for $b$ exists, we have constructed the desired extension.
\item 
Now turn to the Higgs potential	invariant under $A$ and require that, in addition, it be invariant
under the generator $b$ just constructed. Check whether accidental symmetries appear, more on it below.
If they do not, we find a viable 4HDM with the desired non-abelian symmetry group.
\end{itemize}

Note that, when building the scalar potential invariant under the chosen abelian group $A$,
we implicitly assume that the rephasing-insensitive part $V_0$ is always present, where
\begin{eqnarray}
V_0 &=& \sum_{i=1}^4 \left[m_{ii}^2\fdf{i}{i} + \Lambda_{ii}\fdf{i}{i}^2 \right]+ 
\sum_{i < j} \left[\Lambda_{ij}\fdf{i}{i}\fdf{j}{j} + \tilde \Lambda_{ij}\fdf{i}{j}\fdf{j}{i}\right]\,.\label{V0-general}
\end{eqnarray}
When we construct non-abelian extensions, which involve not only rephasing but also permutations,
we will require that the coefficients of $V_0$ are such that $V_0$ is invariant under those permutation as well.
The exact set of relations depends on the specific set of permutations used; the full list of options
can be found in our previous paper \cite{Shao:2023oxt}.

\subsection{Detecting accidental symmetries}

	In the last step of the strategy, we mentioned the necessity of checking whether an accidental symmetry appears
as a result on requiring that the potential be invariant under the symmetry $b$.
This is a subtle issue which requires explanation.
First, when we begin with a model with a rephasing abelian symmetry group $A$, we write the most general potential
invariant under it, in any of the possible implementations of $A$.
The absence of an accidental symmetry is guaranteed by the Smith normal form methods itself described in \cite{Ivanov:2011ae}.
However, by adding $b$ which mixes doublets and requiring the potential to be invariant under it, 
we may inadvertently render the potential invariant under other symmetries, either a rephasing symmetry or 
an additional doublet-mixing one.

Detecting an accidental rephasing symmetry is straightforward. 
It follows from the general analysis of rephasing transformations 
\cite{Ivanov:2011ae} that, if an NHDM scalar potential contains less than $N-1$ rephasing-sensitive terms,
it acquires a continuous rephasing symmetry. Thus, if we want to avoid accidental rephasing continuous symmetries 
in the 4HDM scalar sector, the potential must have at least three rephasing-sensitive terms.
Thus, if imposition of $b$ forces are to set some of the coefficients to zero and reduces 
the number of rephasing sensitive terms to two or less,
the potential will automatically acquire an accidental continuous rephasing symmetry. 
We will encounter such examples at several occasions below.

It is a harder task to detect an accidental symmetry which is neither a pure rephasing 
nor a simple permutation but mixes doublets in a non-trivial way.
A strong hint at the absence of such an accidental symmetry comes from the structure of coefficients 
of the potential. For example, if two quadratic coefficients in Eq.~\eqref{V0-general} are distinct,
such as $m_{11}^2 \not = m_{22}^2$, no symmetry exists mixing the two doublets.
If $m_{11}^2 = m_{22}^2$ and, furthermore, if $\Lambda_{11} = \Lambda_{22}$,
we still can rule out any symmetry mixing $\phi_1$ and $\phi_2$ if the quartic 
cross term coefficient $\Lambda_{12} \not = 2 \Lambda_{11}$.
It is on the basis of such checks that we claim, in each case, the absence of accidental symmetries.

Although such checks provide strong evidence, they do not represent a rigorous proof.
Accidental symmetries can be unambiguously ruled out with basis-independent methods, 
as it was done for the 2HDM and 3HDM.
Such methods have not yet been developed for the 4HDM, although the first steps were done, 
for example, in \cite{Plantey:2024gju}.
Once they appear, they can be applied to prove the absence of accidental symmetries in each case.


\section{Extensions based on $\Z_2\times \Z_2$}\label{section-Z2Z2}

The symmetry group $\Z_2\times \Z_2 = \{e, a_1, a_2, a_1a_2 \} $ is easy to implement in a multi-Higgs model 
as its transformations can be defined as sign flips of certain doublets. 
Within the 3HDM, this construction is unique: the two generators of $\Z_2\times \Z_2$ can be always brought to the form
\begin{equation}
\mbox{3HDM:}\quad a_1 = 
\begin{pmatrix}
	-1 & 0 & 0 \\
	0 & -1 & 0 \\
	0 & 0 & 1 \\
\end{pmatrix}\,,\quad
a_2 = 
\begin{pmatrix}
	1 & 0 & 0 \\
	0 & -1 & 0 \\
	0 & 0 & -1 \\
\end{pmatrix}\,.\quad
\label{Z2Z2-1}
\end{equation}
Since the overall sign flip of all doublets has no effect on the model,
one can also view $a_1$ as the sign flip of $\phi_3$ alone and $a_2$ as the sign flip of $\phi_1$ alone.
This symmetry group was used in the famous Weinberg's model \cite{Weinberg:1976hu} which triggered 
an intense exploration of models with more than two scalar doublets.

With four Higgs doublets, we encounter two non-equivalent implementations of $\Z_2 \times \Z_2$,
which differ by the presence or absence of a 2D invariant subspace. 
We label these two implementations as
\begin{eqnarray}
\mbox{fully represented $\Z_2\times\Z_2$:} && a_1 = \diag(-1,-1, 1, 1)\,,\quad a_2 = \diag(1,-1, -1, 1)\,,\label{extension-Z2Z2-1}\\
\mbox{$\Z_2\times\Z_2$ with a 2D inv. subspace:} && 
a_1 = \sqrt{i}\cdot\diag(-1,1, 1, 1)\,,\ 
a_2 = \sqrt{i}\cdot\diag(1,-1, 1, 1)\,.\label{extension-Z2Z2-2}
\end{eqnarray}
We call the first option as the ``fully represented $\Z_2\times\Z_2$'' because the four doublets
transform as the four distinct singlets of $\Z_2\times\Z_2$:
\begin{equation}
\phi_1 \sim 1_{-+}\,, \quad
\phi_2 \sim 1_{--}\,, \quad
\phi_3 \sim 1_{+-}\,, \quad
\phi_4 \sim 1_{++}\,.
\end{equation}
Using the code {\tt 4HDM Toolbox} \cite{TheCode} to perform exhaustive search as discussed in \cite{Shao:2023oxt}, 
we verified that any $\Z_2 \times \Z_2$ symmetry group 
in the 4HDM scalar sector can indeed be represented by one of the two above options.
Notice that if one removes the fourth doublet, then the two implementations of $\Z_2 \times \Z_2$ 
lead to the same group of the 3HDM, up to the overall phase change. 
Thus, it is the presence of the fourth doublet which distinguishes the two options, despite the fact  
that $\varphi_4$ transforms trivially under $\Z_2 \times \Z_2$. 

Turning now to the automorphism group $\Aut(\Z_2\times\Z_2) \simeq S_3$, where $S_3$ acts on the three non-trivial elements
of $\Z_2\times\Z_2$ by permutations, 
we notice that all three options for the group extension
\begin{equation}
(\Z_2\times\Z_2)\rtimes \Z_2 \simeq D_4\,, \quad
(\Z_2\times\Z_2)\rtimes \Z_3 \simeq A_4\,, \quad
(\Z_2\times\Z_2)\rtimes S_3 \simeq S_4
\end{equation}
were already available for the 3HDM model building \cite{Ivanov:2012fp}.
The extension by $\Z_2$ can be defined by an order-2 transformation $b$, 
which sends $a_1 \mapsto a_2$ and $a_2 \mapsto a_1$.
The extension by $\Z_3$ involves an order-3 transformation $c$, which generates
cyclic permutations of the three elements such as $a_1\mapsto a_2\mapsto a_1a_2 \mapsto a_1$.
The extension by $S_3$ involves both $b$ and $c$.

All three extensions can be readily exported to the 4HDM with the fully represented implementation of $\Z_2\times\Z_2$.
However the other implementation \eqref{extension-Z2Z2-2} only admits the extension $(\Z_2\times\Z_2)\rtimes \Z_2 \simeq D_4$.
Trying to extend it by $\Z_3$ would lead to the equation $c^{-1}a_1 c = a_2$ and $c^{-1}a_2 c = a_1a_2$.
While the former equation can be solved, the latter one has no non-trivial solution because $a_2$ contains the prefactor $\sqrt{i}$
while $a_1a_2$ does not.
Thus, although this peculiar realization of the $\Z_2\times\Z_2$ symmetry leaves a 2D subspace invariant,
its non-abelian extensions can only be $D_4$, although this $D_4$ is different from the group $D_4$ we obtained in the fully represented case. 
Even if one considers non-split extensions by $\Z_2$,
one still arrives only at $D_4$ and not at the quaternion group $Q_4$ because $\Z_2\times\Z_2 \not \subset Q_4$. 

The scalar potential invariant under each implementation of the $\Z_2\times\Z_2$ and its extensions
are rather lengthy but can be readily written out, as they borrow their structures from the 3HDM case.

\section{Extensions based on $\Z_4\times\Z_2$}\label{section-Z4Z2}

\subsection{The three implementations of $\Z_4\times\Z_2$ in the 4HDM}

The abelian group $\Z_4\times\Z_2$ was not available in the 3HDM but appears in the 4HDM, and the analysis of its extensions is much more involved.
Let us begin by describing the group itself, its automorphism group $\Aut(\Z_2\times\Z_4) \simeq D_4$, 
and the conjugacy classes of the subgroups of $D_4$.

The group $\Z_4\times\Z_2$ is generated by $a_1$ of order 4 and $a_2$ of order 2, which commute with each other.
This group contains, among its subgroups, both $\Z_4$ and $\Z_2\times\Z_2$.
We already know from \cite{Shao:2023oxt} that there are three distinct 4HDM implementations of $\Z_4$.
We also established in the previous section that there exist two nonequivalent implementations of $\Z_2\times\Z_2$.
Thus, it is natural to explore implementations of $\Z_4\times\Z_2$ by combining these choices.

It turns out, however, that not all combinations lead to viable models.
For example, let us take the fully represented $\Z_4$ and the fully represented $\Z_2\times\Z_2$.
This can be done by choosing the following generators of the $\Z_4\times\Z_2$:
\begin{equation}
a_1 = \sqrt{i}\cdot\diag(i, -1, -i, 1)\,, \quad a_2 = \diag(1, -1, -1, 1)\,.\label{extension-Z4Z2-invalid}
\end{equation}
One immediately sees that $a_1^2$ and $a_2$ are the pair which gives the fully generated $\Z_2\times\Z_2$.
Next, we take the corresponding $\Z_4$-invariant potential from \cite{Shao:2023oxt} 
and leave only those terms which remain invariant under $a_2$:
\begin{equation}
\begin{aligned}
	\tilde V & = \lambda_3\fdf{1}{3}^2 + \lambda_6\fdf{2}{4}^2 + \lambda_7\fdf{1}{2}\fdf{4}{3} + \lambda_8\fdf{1}{3}\fdf{4}{2} \\
	& + \lambda_9\fdf{1}{3}\fdf{2}{4} + \lambda_{10}\fdf{1}{4}\fdf{2}{3} + h.c.
\end{aligned}
\label{V1-Z4Z2-invalid}
\end{equation}
We remind the reader that, in addition to these terms, the full potential also includes 
the rephasing-insensitive part $V_0$ given in \eqref{V0-general}.
However, the resulting potential $V_0 + \tilde V$ possesses an accidental $U(1)$ symmetry:
\begin{equation}
U(1)_{\rm acc.} = \diag\left(e^{i\alpha}, e^{-i\alpha}, e^{i\alpha}, e^{-i\alpha}\right)\,.\label{U1-accidental}
\end{equation}
Thus, this situation cannot be classified as a $\Z_4\times\Z_2$ symmetric model.

After studying all the combinations of the $\Z_4$ and $\Z_2\times\Z_2$ generators and verifying the results with the code
{\tt 4HDM Toolbox} \cite{TheCode}, we found that only three
nonequivalent options for the $\Z_4\times\Z_2$-symmetric 4HDM exist. We list below there potential and generators:
\begin{eqnarray}
\mbox{option 1:}&& V_1 = \lambda_1\fdf{1}{3}^2 + \lambda_2\fdf{2}{4}^2 + \lambda_3\fdf{1}{2}\fdf{3}{2} + \lambda_4\fdf{1}{4}\fdf{3}{4} + h.c. \nonumber\\
&& a_1^{(1)} = \sqrt{i}\cdot\diag(i,-1,-i,1)\,,\quad a_2^{(1)} = \sqrt{i}\cdot\diag(1,-1,1,1)\,, \label{extension-Z4Z2-a1}\\[2mm]
\mbox{option 2:}&& V_2 = \lambda_1\fdf{1}{2}^2 + \lambda_2\fdf{3}{4}^2 + \lambda_3\fdf{1}{3}\fdf{2}{4} + \lambda_4\fdf{1}{4}\fdf{2}{3}  + h.c.\nonumber\\
&& a_1^{(2)} = \diag(i,i,-1,1)\,,\quad a_2^{(2)} = \diag(-1,1,-1,1)\,,  \label{extension-Z4Z2-a2}\\[2mm]
\mbox{option 3:}&& V_3 = \lambda_1\fdf{1}{2}^2 + \lambda_2\fdf{1}{3}^2 + \lambda_3\fdf{2}{3}^2 + \lambda_4\fdf{1}{4}\fdf{3}{4}  + h.c.\nonumber\\
&& a_1^{(3)} = i^{3/4}\cdot\diag(i,i,-i,1)\,,\quad a_2^{(3)} = \diag(-1,1,-1,1)\,.  \label{extension-Z4Z2-a3}
\end{eqnarray}
Below we will build non-abelian extensions for each of these three options.

Let us stress once more that if an NHDM scalar potential contains less than $N-1$ rephasing-sensitive terms,
it acquires a continuous rephasing symmetry. 
Thus, if we want to avoid accidental rephasing symmetries in the 4HDM scalar sector,
the potential must have at least three rephasing-sensitive terms.
Each of the three $\Z_4\times\Z_2$-symmetric options shown above contains four terms,
and no continuous symmetry is present.
However if a specific extension requires that any two of these coefficients vanish,
the potential will automatically acquire an accidental continuous rephasing symmetry. 

\subsection{The automorphism group of $\Z_4\times\Z_2$ and its conjugacy classes}

The automorphism group of $\Z_4\times\Z_2$ is 
\begin{equation}
\Aut(\Z_4\times\Z_2)\simeq D_4 = \langle b, c\,|\, b^4 = c^2 = e,\, c b c = b^{-1}\rangle\,.
\end{equation}
The two automorphisms $b$ and $c$, which generate $\Aut(\Z_4\times\Z_2)$, act on the group $\Z_4\times\Z_2$ in the following way:
\begin{eqnarray}
b: \ \left\{
\begin{aligned}
	&a_1 \mapsto a_1 a_2\\
	&a_2 \mapsto a_2a_1^2
\end{aligned}
\right.,
\qquad 
\tilde b \equiv b^2: \ \left\{
\begin{aligned}
	&a_1 \mapsto a_1^{-1}\\
	&a_2 \mapsto a_2
\end{aligned}
\right.,
\qquad 
c: \ \left\{
\begin{aligned}
	&a_1 \mapsto a_1\\
	&a_2 \mapsto a_2a_1^2
\end{aligned}
\right.
\label{extension-Z4Z2-bc}
\end{eqnarray}
Next, we need to list all subgroups of $D_4$, and also arrange them into conjugacy classes. 
We remind the reader that $D_4$ (the symmetry group of the square) contains only one 
subgroup $\Z_4$ (the rotations of the square) but five subgroups $\Z_2$ (reflections of the square).
These five $\Z_2$'s form three conjugacy classes:
two reflections parallel to the sides (generated by $c$ or $b^2 c$), 
two reflections along the diagonals (generated by $bc$ or $b^3 c$), 
and the unique point reflection $\tilde b$, which commutes with 
any symmetry of the square and, group-theoretically, corresponds to the center of $D_4$.
It is straightforward to check that there are also two $\Z_2 \times \Z_2$ subgroups, each forming its own conjugacy class.
Thus, we obtain eight proper non-trivial subgroups of $D_4$, which are explicitly listed in Table~\ref{table-subgroups-D4}.

\begin{table}[H]
\centering
\begin{tabular}[t]{cc|ccc}
	\toprule
	conjugacy classes & groups & option 1 & option 2 & option 3 \\
	\midrule
	$\Z_2$ & $\{e,\,b^2 \}$ & $\checkmark$ & $\checkmark$ & {\gray $\times$} \\[1mm]
	$\Z_2$ & $\{e,\,b^2c\}$ & $\checkmark$ & $\checkmark$ & {\gray $\times$} \\
	 	   & $\{e,\,c\}$ & $\checkmark$ & $\checkmark$ & {\gray $\times$} \\[1mm]
	$\Z_2$ & $\{e,\,bc\}$ & {\gray $\times$} & $\checkmark$ & $\checkmark$ \\
	       & $\{e,\,b^3c\}$ & {\gray $\times$} & $\checkmark$ & {\gray $\times$} \\[1mm]
	$\Z_4$ & $\{e,\,b,\,b^2,\,b^3\}$ & {\gray $\times$} & $\checkmark$ & {\gray $\times$} \\[1mm]
	$\Z_2 \times \Z_2$ & $\{e,\,b^2,\,c,\,b^2c \}$ & $\checkmark$ & $\checkmark$ & {\gray $\times$} \\[1mm]
	$\Z_2 \times \Z_2$ & $\{e,\,b^2,\,bc,\,b^3c \}$ & {\gray $\times$} & $\checkmark$ & {\gray $\times$} \\
	\bottomrule
\end{tabular}
\caption{The non-trivial proper subgroups of $D_4$ arranged by the conjugacy classes.
For each subgroup, the extensions are possible (marked with $\checkmark$)
or impossible ({\gray $\times$}) to construct, depending on which option we choose out of the three possible options 
for the $\Z_4\times \Z_2$ group, given in Eqs.~\eqref{extension-Z4Z2-a1}--\eqref{extension-Z4Z2-a3}.
}
\label{table-subgroups-D4}
\end{table}

As we will find below, not all subgroups can be used to build extensions. 
Moreover, the ability of a subgroup of $D_4$ to produce a non-abelian extension by $\Z_4\times \Z_2$ depends 
on the particular implementation of $\Z_4\times \Z_2$ in the 4HDM, given by Eqs.~\eqref{extension-Z4Z2-a1}--\eqref{extension-Z4Z2-a3}.
In what follows, we will first select this implementation (options 1, 2, or 3) and then go through 
the list of subgroups of $D_4$, trying each time to build an extension.
To help the reader navigate through the rather laborious procedure, we summarize in the same Table~\ref{table-subgroups-D4}
which subgroups for which option yield a non-trivial extension.

\subsection{Building $\Z_4\times\Z_2$ extensions: option 1}

\subsubsection{Three choices of $\Z_2$ extensions with no other extensions available}

We begin constructing extensions for the first implementation of the group $\Z_4\times \Z_2$, which is given in Eqs.~\eqref{extension-Z4Z2-a1}.
For the reader's convenience, we repeat it here:
\begin{eqnarray}
	\mbox{option 1:}&& V_1 = \lambda_1\fdf{1}{3}^2 + \lambda_2\fdf{2}{4}^2 + \lambda_3\fdf{1}{2}\fdf{3}{2} + \lambda_4\fdf{1}{4}\fdf{3}{4} + h.c. \nonumber\\
	&& a_1 = \sqrt{i}\cdot\diag(i,-1,-i,1)\,,\quad a_2 = \sqrt{i}\cdot\diag(1,-1,1,1)\,, \label{extension-Z4Z2-a1-again}
\end{eqnarray}
The prefactors $\sqrt{i}$ in $a_1$ and $a_2$ make it clear that the equation of the form $p^{-1}a_1p = a_1a_2\cdot i^r$ does not have
any solution for $p$ for any integer $r$. In general, one must have either an even number of $a_i$'s or an odd number of $a_i$'s simultaneously on the left-hand and the right-hand sides of this equation. Looking at the definitions of the automorphisms $b$ and $c$ in Eq.~\eqref{extension-Z4Z2-bc},
we see that $c$ and even powers of $b$ can be used, while the automorphisms $b$, $b^3$, $b^3c$, $bc$ cannot be represented as $PSU(4)$ transformations. 
It is this reasoning which allows us to cross out half of the subgroups from Table~\ref{table-subgroups-D4}, option 1.

Another remark concerns the possible relation of the second and the third lines of Table~\ref{table-subgroups-D4}.
These two $\Z_2$ subgroups, $\lr{b^2c}$ and $\lr{c}$, belong to one conjugacy class and are conjugate inside $D_4$.
That is, the equation $q^{-1} c q = b^2 c$ has a solution --- in fact, four solutions --- inside $D_4$: 
$q = b,\, b^3,\, bc,\, b^3c$. However these are precisely the elements of $D_4$ which cannot be represented by $PSU(4)$ transformations.
Therefore, we will not be able to use the short-cut argument described at the end of Section~\ref{subsection-theorem}
and will need to consider the second and third lines of Table~\ref{table-subgroups-D4} separately.

\subsubsection{Using the first $\Z_2$}

Consider the first $\Z_2$ group from Table~\ref{table-subgroups-D4}, which is generated by $\tilde b \equiv b^2$.
The action of $\tilde b$ is defined in Eqs.~\eqref{extension-Z4Z2-bc}. The corresponding matrix in $SU(4)$
must satisfy the following equations: 
\begin{equation}
\left\{
\begin{aligned}
	& \tilde b^{-1}a_1 \tilde b = a_1^3\cdot i^{r_1}\\
	& \tilde b^{-1}a_2 \tilde b = a_2\cdot i^{r_2}
\end{aligned}
\right.
\quad \Rightarrow \quad
\left\{
\begin{aligned}
	& a_1 \tilde b = \tilde b a_1^3 \cdot i^{r_1}\\
	& a_2 \tilde b = \tilde b a_2 \cdot i^{r_2}
\end{aligned}
\right.\,.\label{extension-Z4Z2-1-eq1}
\end{equation}
Here, $r_1$ and $r_2$ are arbitrary integers.
The presence of factors $i^{r_i}$ reflects the fact that,
although we work with matrices from $SU(4)$, all equalities are defined modulo to the center of $SU(4)$,
that is, modulo to powers of $i$, see details in \cite{Shao:2023oxt}.
Using the methods developed in \cite{Shao:2023oxt}, we solve this system of linear matrix equations
and find that $\tilde b$ must be of the form
\begin{equation}
\tilde b = 
\begin{pmatrix}
	0 & 0 & b_{13} & 0 \\
	0 & b_{22} & 0 & 0 \\
	b_{31} & 0 & 0 & 0 \\
	0 & 0 & 0 & b_{44} 
\end{pmatrix}\,,
\label{extension-Z4Z2-1-eq2}
\end{equation}
where all entries are pure phase factors. Their phases are not constrained by group theory
but can be determined from the phases of the complex coefficients of the corresponding potential $V_1$.

Knowing the action of $\tilde b$ on the generators $a_1$ and $a_2$ does not specify the extension 
$G = (\Z_4 \times \Z_2)\,.\,\Z_2$ uniquely.
We still need to define the value of $\tilde b^2$ inside $G$. One choice is to set $\tilde b^2 = e$, 
which results in a split extension,
also called the semi-direct product: $(\Z_4 \times \Z_2)\rtimes \Z_2 \simeq D_4 \times \Z_2$.
In this expression, the $D_4$ factor is generated by $a_1$ and $\tilde b$, while the last $\Z_2$ factor is the same 
$\Z_2$ subgroup generated by $a_2$.

Consider now the potential $V_1$ in Eq.~\eqref{extension-Z4Z2-a1-again}.
Upon a suitable rephasing of each doublet, we can set the coefficient $\lambda_1$ real. Note that we classify potentials with discrete symmetry up to basis changes, meaning if two potentials are related by a basis change, we only pick one as representative. This is supported by the theorem we proved in Section \ref{subsection-theorem}. Rephasing is clearly a basis change, so we have the freedom to choose a basis in which the matrix form of the transformation $b$ looks nice and simple. This basis change trick applies to subsequent texts too. We draw the reader's attention again on the difference between freedom of basis change and actual symmetry of a potential, which is the invariance of a potential under certain transformation, which can easily lead to confusion.
In the basis which $\lambda_1$ is real, we observe that 
\begin{equation}
\tilde b = \sqrt{i}
\begin{pmatrix}
	0 & 0 & 1 & 0 \\
	0 & -1 & 0 & 0 \\
	1 & 0 & 0 & 0 \\
	0 & 0 & 0 & -1 \\
\end{pmatrix}
\label{Z4xZ2_Option1_by_Z2_Z2xD4}
\end{equation}
becomes a generator of a symmetry of the potential without any additional requirement on the coefficients of $V_1$.
In other words, if we implement the group $\Z_4\times \Z_2$ using option 1, 
$V_1$ acquires an accidental {\em discrete} symmetry,
so that its total symmetry content is automatically enhanced to $D_4 \times \Z_2$.
Thus, the only condition for a $\Z_4\times\Z_2$ model, option 1, to become invariant under $D_4 \times \Z_2$
is that the rephasing-insensitive part $V_0$ is invariant under $\phi_1 \leftrightarrow\phi_3$.

Having obtained the generic form of $\tilde b$ in Eq.~\eqref{extension-Z4Z2-1-eq2}, 
we can also assume that $\tilde b^2 \not = e$ 
but instead lies inside $\Z_4 \times \Z_2$. 
Since the elements $(\tilde b^2)_{11} = (\tilde b^2)_{33}$, the only available options for $\tilde b^2$ are 
$a_2$, $a_1^2$, or $a_1^2a_2$. All three choices will result in non-split extensions.
Under the first choice $\tilde b^2 = a_2$, the generator $\tilde b$ takes, 
in a suitable real-$\lambda_1$ basis, the following form:
\begin{equation}
\tilde b = i^{1/4}
\begin{pmatrix}
	0 & 0 & 1 & 0 \\
	0 & i & 0 & 0 \\
	1 & 0 & 0 & 0 \\
	0 & 0 & 0 & 1 \\
\end{pmatrix}\,.
\label{Z4xZ2_Option1_by_Z2_Z2xQ4}
\end{equation}
Imposing this symmetry on $V_1$ flips the signs of the $\lambda_2$ and $\lambda_3$ terms. 
Thus, we are forced to set $\lambda_2 = \lambda_3 = 0$. But then we have too few rephasing sensitive terms,
and an accidental $U(1)$ symmetry emerges.
Similarly, under the second choice $\tilde b^2 = a_1^2$, we find that we are forced 
to set $\lambda_3 = \lambda_4 = 0$. 
The last choice $\tilde b^2 = a_1^2a_2$ leads to the vanishing $\lambda_2$ and $\lambda_4$. 
Thus, all attempts of construct a model based on a finite non-split extension of the first $\Z_2$ subgroup fail.

The bottom line is: extending $\Z_2 = \{e, \tilde b\}$ by $\Z_4 \times \Z_2$, implemented as in Eq.~\eqref{extension-Z4Z2-a1-again}, 
is only possible for the split extension $(\Z_4 \times \Z_2)\rtimes \Z_2 \simeq D_4 \times \Z_2$.
The extra generator is the permutation $\tilde b$ given by Eq.~\eqref{Z4xZ2_Option1_by_Z2_Z2xD4}.
The only condition for this symmetry group is that $V_0$ is invariant under $\phi_1 \leftrightarrow\phi_3$.
The $V_1$ part of the potential is automatically invariant under $\tilde b$.

	\subsubsection{Using the second $\Z_2$}

Next, we pick $\Z_2\simeq \langle b^2c\rangle$ and denote $d = b^2c$. The corresponding equations are
\begin{equation}
	\left\{
	\begin{aligned}
		& d^{-1}a_1 d = a_1^3\cdot i^{r_1}\\
		& d^{-1}a_2 d = a_2a_1^2\cdot i^{r_2}
	\end{aligned}
	\right.
	\quad \Rightarrow \quad
	\left\{
	\begin{aligned}
		& a_1 d = d a_1^3 \cdot i^{r_1}\\
		& a_2 d = d a_2a_1^2 \cdot i^{r_2}
	\end{aligned}
	\right.\,,\label{extension-Z4Z2-1-eq2-1}
\end{equation}
which have solutions only for $r_1 = 1$, $r_2 = 1$:
\begin{equation}
	d = 
	\begin{pmatrix}
		d_{11} & 0 & 0 & 0 \\
		0 & 0 & 0 & d_{24} \\
		0 & 0 & d_{33} & 0 \\
		0 & d_{42} & 0 & 0 
	\end{pmatrix}\,.
	\label{extension-Z4Z2-1-eq2-2}
\end{equation}
In terms of the Higgs doublets, this transformation permutes $\phi_2\leftrightarrow \phi_4$. 
Since $d^2$ has the form $\diag(x, y, z, y)$, it can be either $e$ or 
a non-trivial element of $\Z_4\times\Z_2$, namely, $a_1^2$, $a_1a_2$, $a_1^3a_2$. 

We first consider the split extension by setting $d^2 = e$. The non-abelian group obtained in this way is labeled as
\begin{equation}
G = (\Z_4 \times \Z_2)\rtimes \Z_2 = {\tt SmallGroup(16,13)}\,\label{G1}
\end{equation}
where we indicated the {\tt GAP} id of this group of order 16. 
This group is also known as the Pauli group $G_1$ generated by the three Pauli matrices under multiplication.
In group theoretic terms, it can be also defined as the central product of $D_4$ and $\Z_4$, 
denoted by $\Z_4\circ D_4$. 
The transformation $d$, in a suitable basis, has the form
\begin{equation}
	d = \sqrt{i}
	\begin{pmatrix}
		-1 & 0 & 0 & 0 \\
		0 & 0 & 0 & 1 \\
		0 & 0 & -1 & 0 \\
		0 & 1 & 0 & 0 \\
	\end{pmatrix}\,.
	\label{extension-Z4Z2-eq2-3}
\end{equation}
We select a basis with real $\lambda_2$, so the invariance under Eq.~\eqref{extension-Z4Z2-eq2-3} only requires $\lambda_3 = \lambda_4$,
together with the constraints in $V_0$ obtained by imposing the invariance under $\phi_2\leftrightarrow\phi_4$. 

We can also try to build non-split extensions by solving matrix equation $d^2 = i^r\cdot a_1^2$. 
A solution exists for $r=-1$ 
and has the following form:
\begin{equation}
	d = 
	\begin{pmatrix}
		i & 0 & 0 & 0 \\
		0 & 0 & 0 & 1 \\
		0 & 0 & i & 0 \\
		0 & 1 & 0 & 0 \\
	\end{pmatrix}\,.
	\label{extension-Z4Z2-eq2-4}
\end{equation}
Invariance under this $d$ requires $\lambda_3 = -\lambda_4$ in the real $\lambda_2$ basis. 
This construction also leads to the same group $G_1$ as in Eq.~\eqref{G1}, 
but we arrived at it using non-split extension procedure.
This is not a coincidence: a potential with a real $\lambda_2$ and $\lambda_3 = -\lambda_4$ 
can be transformed into a potential with another real $\lambda_2$ and $\lambda_3 = \lambda_4$
by the basis change $\phi_{1,2,3} \mapsto \phi_{1,2,3}$, $\phi_4 \mapsto i \phi_4$.
Alternatively, we can note that a potential invariant under $d$ from Eq.~\eqref{extension-Z4Z2-eq2-4}
is also invariant under $d' = d a_2$, whose square is $(d')^2 = -i\cdot \id_4$.
Thus, the group $G = A\,.\,\Z_2$ with $\Z_2 = \lr{d}$ can also be represented as 
$G = A\rtimes\Z'_2$, where $\Z'_2 = \lr{d'}$.
In short, the attempt at a non-split extension leads us to the same result as the split extension
because the two models are related by a mere basis change and represent the same physical situation.

We also checked that the other non-split extension attempts, $d^2 = i^r\cdot a_1a_2$ and $d^2 = i^r\cdot a_1^3a_2$, 
lead to continuous symmetries and are not realizable as 4HDM discrete symmetries.

\subsubsection{Using the third $\Z_2$}

Next, let us take the $\Z_2$ group $\{e,c\}$ from the third line in Table~\ref{table-subgroups-D4}.
The action of $c$ is defined in Eq.~\eqref{extension-Z4Z2-bc}, and the corresponding equations are
\begin{equation}
\left\{
\begin{aligned}
	& c^{-1}a_1 c = a_1\cdot i^{r_1}\\
	& c^{-1}a_2 c = a_2a_1^2\cdot i^{r_2}
\end{aligned}
\right.
\quad \Rightarrow \quad
\left\{
\begin{aligned}
	& a_1 c = c a_1 \cdot i^{r_1}\\
	& a_2 c = ca_2a_1^2 \cdot i^{r_2}
\end{aligned}
\right.\,.\label{extension-Z4Z2-1-eq3}
\end{equation}
This system has solutions only for $r_1 = 2$, $r_2 = 1$, leading to
\begin{equation}
c =
\begin{pmatrix}
	0 & 0 & c_{13} & 0 \\
	0 & 0 & 0 & c_{24} \\
	c_{31} & 0 & 0 & 0 \\
	0 & c_{42} & 0 & 0 \\
\end{pmatrix}\,,
\label{extension-Z4Z2-1-eq4}
\end{equation}
which permutes $\phi_1\leftrightarrow \phi_3$ and $\phi_2\leftrightarrow\phi_4$ simultaneously.
The split extension, $c^2 = e$, leads to the same Pauli group $G_1$ mentioned above. 
In the basis, where $\lambda_1$ and $\lambda_2$ are real,
the transformation $c$ takes the form
\begin{equation}
c =
\begin{pmatrix}
	0 & 0 & 1 & 0 \\
	0 & 0 & 0 & 1 \\
	1 & 0 & 0 & 0 \\
	0 & 1 & 0 & 0 \\
\end{pmatrix}\,,
\label{extension-Z4Z2-1-eq5}
\end{equation}
The only condition we must impose on the potential $V_1$ in Eq.~\eqref{extension-Z4Z2-a1-again}
is $\lambda_3 = \lambda_4$, which must be accompanied with condition that the rephasing-insensitive part $V_0$ 
be invariant under the simultaneous change $\phi_1 \leftrightarrow \phi_3$ and $\phi_2 \leftrightarrow \phi_4$.
Note that this construction is nearly identical to the second $\Z_2$ extension; 
they differ only in $V_0$, not in $V_1$.

Attempts to build non-split extensions proceeds along the same lines as above. 
The only possibility is to set $c^2 = a_1^2$. In the real $\lambda_1$ and $\lambda_2$ basis, 
it leads to 
\begin{equation}
c = \sqrt{i}
\begin{pmatrix}
	0 & 0 & i & 0 \\
	0 & 0 & 0 & 1 \\
	i & 0 & 0 & 0 \\
	0 & 1 & 0 & 0 \\
\end{pmatrix}\,,
\label{extension-Z4Z2-1-eq6}
\end{equation}
Then, if $\lambda_3 = - \lambda_4$ and if, in addition, $V_0$ is invariant under 
$\phi_1 \leftrightarrow \phi_3$ together with $\phi_2 \leftrightarrow \phi_4$,
the full potential becomes invariant under this $c$.
The total symmetry group is again $G_1$; thus, we recover the split extension in disguise.

\subsubsection{Using $\Z_2 \times \Z_2$}

We have already established that $V_1$ is automatically invariant under $\tilde b$ in Eq.~\eqref{Z4xZ2_Option1_by_Z2_Z2xD4}.
Let us now continue with the above case symmetric under $c$ and combine it with $\tilde b$,
which brings us to the first $\Z_2 \times \Z_2$ group of Table~\ref{table-subgroups-D4}.
In order to achieve this symmetry, we need to impose an additional constraint on $V_0$.
As we saw, $c$ requires invariance of simultaneous permutation $\phi_1\leftrightarrow \phi_3$ and $\phi_2\leftrightarrow \phi_4$,
which is a less stringent constraint than invariance under $\phi_1\leftrightarrow \phi_3$ and $\phi_2\leftrightarrow \phi_4$, separately. 
Indeed, the $\fdf{1}{1}\fdf{2}{2} + \fdf{3}{3}\fdf{4}{4}$ in $V_0$ are invariant under $c$ but not under $\tilde{b}\in \Z_2\times\Z_2$. 

Let us now identify the symmetry group emerging in this case.
With $\tilde b \equiv b^2$ given in Eq.~\eqref{Z4xZ2_Option1_by_Z2_Z2xD4} and $c$ given in Eq.~\eqref{extension-Z4Z2-1-eq5},
we can verify that $\tilde b^2 = c^2 = e$ as well as $[\tilde b, c] = e$.
Thus, we get the split extension
\begin{equation}
	G = (\Z_4\times\Z_2)\rtimes (\Z_2 \times \Z_2) = {\tt SmallGroup(32,49)}\,,\label{extension-Z4Z2-1-extra}
\end{equation}
which is known as the extra-special group of order 32, plus-type, and is labeled as $\mathbf{2^{1+4}_+}$.

It is also possible to build $c$ such that $c^2=e$, but the commutator $[\tilde b,c] \not = e$ although it still lies inside $\Z_4 \times \Z_2$.
This allows us to consider non-split extensions of the form $ (\Z_4\times\Z_2)\,.\,(\Z_2 \times \Z_2)$. 
We checked all choices for $[\tilde b,c]$ and found that many lead to continuous symmetries.
For example, if $[\tilde b,c] = a_1a_2$, then $c$ requires $\lambda_1$ to be imaginary while $\tilde b$ requires
it to be real. Thus, we must set $\lambda_1 = 0$, but in this case the potential $V_1$ acquires the continuous symmetry 
of the form diag$(e^{i\alpha}, 1, e^{-i\alpha}, 1)$.
The net result is that, by combining $\tilde b$ and any $c$ of the form of Eq.~\eqref{extension-Z4Z2-1-eq4}, 
we can only arrive at ${\tt SmallGroup(32,49)}$ through a split or non-split extension procedure. 
However since the group is the same, we end up only at the split extension of
$\Z_2\times \Z_2$ by $\Z_4\times\Z_2$.

This construction wraps up all the extension cases we have with the first option for $\Z_4\times\Z_2$.

\subsection{Building $\Z_4\times\Z_2$ extensions: option 2}\label{subsection-Z4Z4-option2}

Next, we consider the second way the $\Z_4\times \Z_2$ group can be implemented in the 4HDM. 
The potential and the generators $a_1$ and $a_2$ are given in Eqs.~\eqref{extension-Z4Z2-a2};
we repeat them here for the reader's convenience:
\begin{eqnarray}
\mbox{option 2:}&& V_2 = \lambda_1\fdf{1}{2}^2 + \lambda_2\fdf{3}{4}^2 + \lambda_3\fdf{1}{3}\fdf{2}{4} + \lambda_4\fdf{1}{4}\fdf{2}{3}  + h.c.\nonumber\\
&& a_1 = \diag(i,i,-1,1)\,,\quad a_2 = \diag(-1,1,-1,1)\,.  \label{extension-Z4Z2-a2-again}
\end{eqnarray}
Unlike in option 1, these generators do not carry the prefactors $\sqrt{i}$. 
As a result, all automorphisms of the $\Z_4\times\Z_2$ can be represented as $PSU(4)$ transformations.
This simplifies our task as we do not need to consider the subgroups in Table~\ref{table-subgroups-D4}
belonging to the same conjugacy class.

We will now describe the results giving fewer details than before because the methods are the same.
We start again with the first $\Z_2$ from Table~\ref{table-subgroups-D4}.
Following the similar steps, we can solve for the matrix $\tilde b$, which in this case exchanges 
$\phi_1 \leftrightarrow \phi_2$ and $\phi_3 \leftrightarrow \phi_4$.
In the basis of real $\lambda_1$ and $\lambda_2$, this transformation is automatically a symmetry of $V_2$.
Thus, we only need to require $V_0$ to be invariant under these permutations, and in this way we again
obtain the non-abelian group $G = (\Z_4 \times \Z_2)\rtimes \Z_2 = D_4\times \Z_2$.

An attempt to build a non-split extension leads to the condition $\tilde b^2 = a_1^2$. 
But such a $\tilde b$ forces us to set $\lambda_3 = \lambda_4 = 0$, leading to an accidental continuous symmetry.

From the next conjugacy class we select the third $\Z_2$ subgroup $\{e, c\}$.
In the real $\lambda_1$ basis, the matrix $c$ only exchanges $\phi_1 \leftrightarrow \phi_2$. 
We need to require that $\lambda_3 = \lambda_4$ and to make sure that $V_0$ is invariant under $\phi_1 \leftrightarrow \phi_2$.
In this way, we again arrive at the the Pauli Group $G_1 = {\tt SmallGroup(16,13)}$ as in Eq.~\eqref{G1}.
Attempts at non-split extensions do not produce any new options.

Unlike for option 1, the third conjugacy class is now available, lines 4 and 5 of Table~\ref{table-subgroups-D4}.
Following line 4, we choose the subgroup $\{e, bc\}$ to construct extensions. 
The solution for the generator $d = bc$ corresponds, in a suitable basis, to the simultaneous exchange 
$\phi_1 \leftrightarrow \phi_4$ and $\phi_2 \leftrightarrow \phi_3$:
\begin{equation}
	d =  
	\begin{pmatrix}
		0 & 0 & 0 & 1 \\
		0 & 0 & 1 & 0 \\
		0 & 1 & 0 & 0 \\
		1 & 0 & 0 & 0 \\
	\end{pmatrix}\,.
	\label{extension-Z4Z2-2-eq0}
\end{equation}
The conditions for this symmetry to be present are $\lambda_1 = \lambda_2^*$ and
real $\lambda_3$ and $\lambda_4$, plus matching conditions on $V_0$. 
The symmetry group we obtain by combining $\Z_4\times \Z_2$ generated by $a_1$, $a_2$ and
$\Z_2$ generated by $d$ is also a semi-direct product $(\Z_4\times \Z_2)\rtimes \Z_2$ but a different one:
\begin{equation}
G = (\Z_4\times \Z_2)\rtimes \Z_2 = {\tt SmallGroup(16,3)}\,.
\end{equation} 
The non-split extension procedure with $d^2 = a_1^2a_2$ brings us again to this group,
which is thus a split extension in disguise.
Line 5 of Table~\ref{table-subgroups-D4} can be analyzed in a similar way and leads to the same group. 

If $V_0$ is invariant under 
$\phi_1 \leftrightarrow \phi_2$ and, independently, under $\phi_3 \leftrightarrow \phi_4$,
the full symmetry group is further enhanced,
through either a split or a non-split extension procedure.
We found that extension by line 7 of Table~\ref{table-subgroups-D4} leads to the same group {\tt SmallGroup(32,49)}
as in Eq.~\eqref{extension-Z4Z2-1-extra}, while extension by line 8 produces a new option, 
the group denoted as {\tt SmallGroup(32,27)}.

With option 2, we can also use the subgroup $\Z_4$ to build an extension, which was impossible for option 1.
The action of the generator $b$ defined in Eq.~\eqref{extension-Z4Z2-bc} tells us that the symmetry group 
we are going to construct is the following group of order 32:
\begin{equation}
	G = (\Z_4\times\Z_2)\rtimes\Z_4\simeq {\tt SmallGroup(32, 6)}\,,\label{extension-Z4Z2-2-eq1}
\end{equation}
which is also known as the faithful semi-direct product $(\Z_2)^3\rtimes\Z_4$.
The equations $b^{-1}a_1b =a_1 a_2 \cdot i^{r_1}$ and 
$b^{-1}a_2b = a_1^2 a_2 \cdot i^{r_2}$ have the following generic solution:
\begin{equation}
b = 
\begin{pmatrix}
	0 & 0 & b_{13} & 0 \\
	0 & 0 & 0 & b_{24} \\
	0 & b_{32} & 0 & 0 \\
	b_{41} & 0 & 0 & 0 \\
\end{pmatrix}\,,
\label{extension-Z4Z2-2-eq2}
\end{equation}
that is, the cyclic permutation $\phi_1 \mapsto \phi_3 \mapsto \phi_2 \mapsto \phi_4 \mapsto \phi_1$,
possibly corrected by the phase rotations.
Since $b^4$ is proportional to $\id_4$, the extension can only be split.
This transformation is the symmetry of $V_2$ if $\lambda_1 = \lambda_2$ are real
and $\lambda_3 = \lambda_4^*$. In addition, it strongly constrains the rephasing-insensitive part $V_0$
reducing it to
\begin{eqnarray}
V_0 &=& m^2 \left(\fdfn{1}{1} + \fdfn{2}{2} + \fdfn{3}{3} + \fdfn{4}{4}\right)
+ \Lambda \left(\fdfn{1}{1} + \fdfn{2}{2} + \fdfn{3}{3} + \fdfn{4}{4}\right)^2\nonumber\\
&&+\Lambda' \left(\fdfn{1}{1} + \fdfn{2}{2}\right)\left(\fdfn{3}{3} + \fdfn{4}{4}\right) 
+\Lambda'' \left[\fdf{1}{1}\fdf{2}{2} + \fdf{3}{3}\fdf{4}{4}\right]\nonumber\\
&& + \tilde\Lambda' \left(|\fdfn{1}{3}|^2 + |\fdfn{2}{3}|^2 + |\fdfn{1}{4}|^2 + |\fdfn{2}{4}|^2 \right)
+ \tilde\Lambda''\left(|\fdfn{1}{2}|^2 + |\fdfn{3}{4}|^2\right)\,.\label{extension-Z4Z2-2-eq3}
\end{eqnarray}

The final step in our classification of the extensions based on $\Z_4\times \Z_2$ is to use the full $D_4$.
For the split extensions, the resulting group is
\begin{equation}
G = (\Z_4\times\Z_2) \rtimes D_4 \simeq {\tt SmallGroup(64, 138)}\,.\label{extension-Z4Z2-2-D4}
\end{equation}
This is a group of order 64 also known as the unitriangular matrix group of degree 4 over the field $\mathbb{F}_2$, 
denoted as $UT(4,2)$. 
In order to obtain this symmetry group, we just need to impose invariance under $b$ and $c$.
This is achieved by using $V_0$ given in Eq.~\eqref{extension-Z4Z2-2-eq3} and $V_2$ of the form
\begin{equation}
	V_2 = \lambda_1\left[\fdf{1}{2}^2 + \fdf{3}{4}^2 + h.c.\right] + 
	\lambda_3\left[\fdf{1}{3}\fdf{2}{4} + \fdf{1}{4}\fdf{2}{3} + h.c.\right], \label{extension-Z4Z2-2-eq4}
\end{equation}
with all parameters real.
One can describe the full symmetry content of the potential $V_0 + V_2$ as invariance under $a_1$ and $a_2$ 
given in Eq.~\eqref{extension-Z4Z2-a2-again} as well as the following types of permutations: 
$\phi_1 \leftrightarrow \phi_2$, $\phi_3 \leftrightarrow \phi_4$, and the cyclic permutation
$\phi_1 \mapsto \phi_3 \mapsto \phi_2 \mapsto \phi_4 \mapsto \phi_1$.
This potential contains only eight free parameters and leads to remarkably constrained scalar sector of the model.

\subsection{Building $\Z_4\times\Z_2$ extensions: option 3}

Option 3 for implementing the group $\Z_4\times\Z_2$ in the 4HDM, given in Eq.~\eqref{extension-Z4Z2-a3},
is a special one. Due to the presence of the factor $i^{3/4}$ in the definition of $a_1$, 
most of actions defined in Eq.~\eqref{extension-Z4Z2-bc} do not admit solutions.
The only automorphism which can have a solution is $d = bc$, which maps $a_1$ to $a_1a_2$ and keeps $a_2$ unchanged.
Using the same methods as before, we find that the generator $d$ swaps $\phi_1$ and $\phi_3$, 
possibly accompanied with phase factors.

If $d^2 = e$, we deal with the split extension $(\Z_4\times \Z_2)\rtimes \Z_2 = {\tt SmallGroup(16,3)}$,
just as we already encountered in option 2. In order to arrive at this symmetry group, we need to require that the coefficients of $V_3$ 
in Eq.~\eqref{extension-Z4Z2-a3}, in a suitable basis, satisfy $\lambda_1^* = \lambda_3$ and $\lambda_2$ is real.
As for the non-split extensions, the only choice which does not lead to continuous accidental symmetries is 
$d^2 = a_1^2 a_2$, which leads to the same symmetry group ${\tt SmallGroup(16,3)}$.


\section{Extensions based on $\Z_2\times\Z_2\times\Z_2$}\label{section-Z2Z2Z2}

\subsection{$(\Z_2)^3$ as a vector space and its automorphisms}

The rephasing group $\Z_2\times\Z_2\times\Z_2$, or $(\Z_2)^3$ for short, can always be implemented within the 4HDM via sign flips of individual doublets. 
We have three, not four, $\Z_2$ factors just because flipping the signs of the first three doublets is equivalent to flipping the signs of fourth doublet.
The three generators of the group $(\Z_2)^3$ can be selected as
\begin{equation}
a_1 = \sqrt{i}\cdot\diag(-1,1,1,1),\quad a_2 = \sqrt{i}\cdot\diag(1,-1,1,1),\quad a_3 = \sqrt{i}\cdot\diag(1,1,-1,1)\,,\label{E8-eq1}
\end{equation}
where the factors $\sqrt{i}$ are required to guarantee $\det a_i = 1$.
The potential invariant under sign flips includes, in addition to the rephasing-insensitive part $V_0$, the following collection of quartic terms:
\begin{equation}
V_1 = \lambda_{12}\fdf{1}{2}^2 + \lambda_{13}\fdf{1}{3}^2 + \lambda_{23}\fdf{2}{3}^2 
+ \lambda_{14}\fdf{1}{4}^2 + \lambda_{24}\fdf{2}{4}^2 + \lambda_{34}\fdf{3}{4}^2 + h.c.\,,
\label{E8-V1}
\end{equation}
where all the coefficients can be complex.
Using the package {\tt 4HDM Toolbox} \cite{TheCode}, we verified that any implementation of the group $\Z_2\times\Z_2\times\Z_2$
in the 4HDM can indeed be represented, in a suitable basis, by Eq.~\eqref{E8-eq1}.

It is convenient to view the group $(\Z_2)^3$ as a three-dimensional vector space over the finite field with two elements $\mathbb{F}_2 = \{0, 1\}$.
Take a group element $g \in (\Z_2)^3$, which we write in the usual multiplicative notation as $g = a_1^{m_1}a_2^{m_2}a_3^{m_3}$. 
Now think of a three-dimensional space with the basis vectors $\hat{a}_1$, $\hat{a}_2$, $\hat{a}_3$,
which represent the generators $a_1$, $a_2$, and $a_3$. Then $g$ can be associated with the linear combination
$\hat{g} = m_1 \hat{a}_1 + m_2 \hat{a}_2 + m_3 \hat{a}_3$, or simply $(m_1, m_2, m_3)^T$, where the numbers $m_i \in \mathbb{F}_2 = \{0, 1\}$
represent the powers of the three commuting generators. 
It is straightforward to check that the axioms of a vector space are satisfied.

This look at the group $(\Z_2)^3$ as a three-dimensional vector space $(\mathbb{F}_2)^3$ offers a nice way 
to represent and study the automorphisms of $(\Z_2)^3$.
Indeed, an automorphism $f: (\Z_2)^3 \to (\Z_2)^3$ is a rule which maps every $g \in (\Z_2)^3$ to its image $f(g) \in (\Z_2)^3$
and satisfies the usual automorphism axioms.
In the vector space language, this automorphism maps the vector $\hat{g} \in (\mathbb{F}_2)^3$ 
to another vector $f(\hat{g})$ of the same space $(\mathbb{F}_2)^3$. This map is linear and invertible. Therefore, it can be represented 
as a $3\times 3$ matrix acting on $(\mathbb{F}_2)^3$, and its entries must also be from $\mathbb{F}_2$. 
For example, the exchange of generators $a_1 \leftrightarrow a_3$, $a_2 \mapsto a_2$ is indeed an automorphism and is represented by the matrix
\begin{equation}
{\mathfrak b} = 
\begin{pmatrix}
	0 & 0 & 1\\
	0 & 1 & 0\\
	1 & 0 & 0
\end{pmatrix}\,.\label{E8-b-example1}
\end{equation}
Here and below, we use the fraktur letters such as ${\mathfrak b}$ to represent automorphisms as acting 
in the vector space $(\mathbb{F}_2)^3$.

It may happen that an automorphism of $(\Z_2)^3$ can be represented by an $SU(4)$ transformation of the four doublets.
For example, if ${\mathfrak b}$ is given by Eq.~\eqref{E8-b-example1},
we apply the same methods as introduced before and solve the following system of equations
\begin{equation}
	\left\{
	\begin{aligned}
		& b^{-1}a_1b = a_3\cdot i^{r_1}\\
		& b^{-1}a_2b = a_2\cdot i^{r_2}\\
		& b^{-1}a_3b = a_1\cdot i^{r_3}\\
	\end{aligned}
	\right.\;,\quad 
	\mbox{with the solution} \quad 
	b = 
	\begin{pmatrix}
		0 & 0 & b_{13} & 0\\
		0 & b_{22} & 0 & 0\\
		b_{31} & 0 & 0 & 0\\
		0 & 0 & 0 & b_{44}\\
	\end{pmatrix}\ \mbox{for}\ r_i=0\,.\label{E8-b-example2}
\end{equation}
Since automorphisms are invertible, $\det {\mathfrak b} \not = 0$. Then, as we work over the field $\mathbb{F}_2$, 
this determinant can only be equal to one.
This is why the automorphism group of $(\Z_2)^3$, the collection of all matrices ${\mathfrak b}$ with $\det {\mathfrak b} = 1$,
can be written as $GL(3,2) = SL(3,2)$.  
The order of this group is easy to establish. The group $(\Z_2)^3$ contains seven non-trivial elements of order 2. 
When defining an automorphism $f$, we can map $a_1$ to any of these seven, then we map $a_2$ to any of the remaining six,
and finally map $a_3$ to any of the remaining elements barring $f(a_1)f(a_2)$. Thus, 
$|\Aut((\Z_2)^3)| = 7\times 6 \times 4 = 168$.

We draw once more the reader's attention to the distinctions between the two representations we use  in this section to study automorphisms: 
the $GL(3,\mathbb{F}_2)$ representation, in which automorphisms are written as $3\times 3$ matrices and are denoted by fraktur symbols,
and the $PSU(4)$ representation, provided it exists, which views the same automorphism as a $4\times 4$ matrix
acting on the four Higgs doublets, labeled with normal symbols.

Next, we use the database {\ttfamily GAP} to study some properties of $\Aut((\Z_2)^3)$ and its subgroups.
The group itself is labeled as {\ttfamily SmallGroup(168,42)} and has 179 subgroups.
Fortunately, many of these subgroups are conjugate to each other. Thanks to the Theorem proved in Section~\ref{subsection-theorem},
we only need to classify the conjugacy classes of these subgroups and then consider only one example in each class.
Using {\ttfamily GAP}, we found that this group has 13 conjugacy classes of the non-trivial proper subgroups, which we list in Table~\ref{table-subgroup_GL32}.
\begin{table}[H]
\centering
\begin{tabular}[t]{c|cccccccccc}
	\toprule
	\makecell[c]{Representative\\Subgroups} & $\Z_2$ & $\Z_3$ & $\Z_4$ & $\Z_2\times\Z_2$ & $S_3$ & {\color{gray}$\Z_7$} & $D_4$ & $A_4$ & {\color{gray}$\Z_3\rtimes\Z_7$} & $S_4$ \\
	\midrule
	\makecell[c]{Number of \\ conjugacy classes} & 1 & 1 & 1 & 2 & 1 & {\color{gray}1} & 1 & 2 & {\color{gray}1} & 2 \\
	\bottomrule
\end{tabular}
\caption{The conjugacy classes of all non-trivial proper subgroups of $GL(3,2)$. The classes of subgroups which contain $\Z_7$ are not available in the 4HDM (see main text)
	and are shown in gray.}
\label{table-subgroup_GL32}
\end{table}

It turns out that not all automorphisms of the abstract group $(\Z_2)^3$ can be defined when $(\Z_2)^3$ is implemented 
as the symmetry group of the 4HDM as in Eq.~\eqref{E8-eq1}. Let us consider, for example, the automorphism
$a_1 \mapsto a_3$, $a_2 \mapsto a_1$, $a_3 \mapsto a_2a_3$. 
Its matrix $\mathfrak{f}$ and the system of equations for $f$ are
\begin{equation}
\mathfrak{f} = 
\begin{pmatrix}
	0 & 0 & 1\\
	1 & 0 & 0\\
	0 & 1 & 1
\end{pmatrix},\;
\mbox{defining}\;
\left\{
\begin{aligned}
	& f^{-1}a_1f = a_3\\
	& f^{-1}a_2f = a_1\\
	& f^{-1}a_3f = a_2a_3\,.
\end{aligned}
\right.\label{E8-f-fail}
\end{equation}
The map $\mathfrak{f}$ is well defined and has order 7, which can be verified by direct multiplication.
However the system of equations has no solutions for $f$. The obstacle is the last equation: 
even with the freedom of multiplication by integer powers of $i$, the equation $a_3f=fa_2a_3\cdot i^r$ cannot produce invertible matrices $f$.
The root of the problem is that $a_2a_3$ in the right-hand side does not possess the $\sqrt{i}$ factor to match this factor from $a_3$ in the left-hand side.
We conclude that extensions of $\Z_7$ --- and in fact of all groups which contain $\Z_7$ as a subgroup --- are impossible in the 4HDM.  

The lesson we draw from the above example is that, when constructing matrices $\mathfrak{b}$, we can only use four rows,
$(1,1,1)$, $(1,0,0)$, $(0,1,0)$, and $(0,0,1)$, and we must pick up three different ones, in any order. 
In this way, we can construct automorphisms of orders 2, 3 and 4.

Having done this exercise, we found two distinct families of transformations of order 2.
The three transformations
\begin{equation}
	{\mathfrak b}_1' = 
\begin{pmatrix}
	1 & 1 & 1\\
	0 & 0 & 1\\
	0 & 1 & 0
\end{pmatrix}\,,
\quad
	{\mathfrak b}_2' = 
\begin{pmatrix}
	0 & 0 & 1\\
	1 & 1 & 1\\
	1 & 0 & 0
\end{pmatrix}\,,
\quad
	{\mathfrak b}_3' = 
\begin{pmatrix}
	0 & 1 & 0\\
	1 & 0 & 0\\
	1 & 1 & 1
\end{pmatrix}
\label{E8-b'-examples}
\end{equation}
form the first family and share the property that they can be written as squares of order-4 transformations from $\Aut(A)$.
The second family contains six transformations of order 2 such as ${\mathfrak b}$ in Eq.~\eqref{E8-b-example1} 
and products of the type $({\mathfrak b}_1')^{-1}{\mathfrak b}{\mathfrak b}_1'$;
the members of this family cannot be written as a square of any order-4 transformation.

It turns out that all the transformations within each family are linked by some transformations $q \in \Aut(A)$
which can be represented by $\tau \in PSU(4)$. Therefore, these two families are exactly the two conjugacy classes
of physically equivalent models which we described in Section~\ref{subsection-theorem}.
To build the full list of $\Z_2$-based extensions, it suffices to consider one representative transformation
from each family.

As for the transformations of order 3 and order 4, we found that all transformations of the same order can be linked 
by a $PSU(4)$ transformations. Thus, we need to consider only one representative $\Z_3$ and $\Z_4$ group.

\subsection{Building extensions}

Let us begin with extending $\Z_2$. Table~\ref{table-subgroup_GL32} tells us that all $\Z_2$ subgroups are conjugate to each other inside $GL(3,2)$,
but the above discussion suggests that we need to separately consider two representative groups,
which are not linked by any $PSU(4)$ transformation.

For the first example, we can select the automorphism ${\mathfrak b}$ in Eq.~\eqref{E8-b-example1}, 
which leads to $b$ of the form given in Eq.~\eqref{E8-b-example2}.
Squaring $b$, we get the diagonal matrix of the type $b^2=\diag(x,y,x,z)$, which can be proportional to $\mathbf{1}_4,\;  a_2,\;  a_1a_3,\;  a_1a_2a_3$.
The choice $b^2 = \mathbf{1}_4$ leads to the split extension $\Z_2\times D_4$. 
In a suitable basis, $b$ corresponds to the exchange $\phi_1 \leftrightarrow \phi_3$
and a sign flip, for example, of $\phi_2$. In order for $V_1$ in Eq.~\eqref{E8-V1} to be invariant under this transformation,
we require, in the real $\lambda_{13}$  basis, that $\lambda_{12} = \lambda_{23}^*$ and $\lambda_{14} = \lambda_{34}$.
Also, the rephasing-insensitive part of the potential, $V_0$, must be invariant under the exchange $\phi_1 \leftrightarrow \phi_3$.

For a non-split extension, we can select $b$ as in Eq.~\eqref{Z4xZ2_Option1_by_Z2_Z2xQ4}, so that $b^2 = a_2$.
In the real $\lambda_{13}$ basis, $V_1$ acquires this symmetry if $\lambda_{12} = -\lambda_{23}^*$, $\lambda_{14} = \lambda_{34}$,
and in addition $\lambda_{24} = 0$. The conditions for the other two non-split extensions can be immediately constructed.
In all three cases of non-split extensions, the total symmetry group is {\tt SmallGroup(16,3)}, which we already encountered
in Section~\ref{subsection-Z4Z4-option2} when extending $\Z_2$ by $\Z_4 \times \Z_2$.
We did not check whether these two versions of the 4HDM invariant under {\tt SmallGroup(16,3)} can be related by 
a basis change and lead to the same phenomenology.

The second $\Z_2$ example is ${\mathfrak b}_2'$ in Eq.~\eqref{E8-b'-examples}, which corresponds to 
the simultaneous transformation $\phi_1 \leftrightarrow \phi_3$ and $\phi_2 \leftrightarrow \phi_4$.
Clearly, it is equivalent to the cyclic permutation $\phi_1 \mapsto \phi_2 \mapsto \phi_3 \mapsto \phi_4 \mapsto \phi_1$ applied twice.
In this case, we can only have a split extension, and the total symmetry group is $\Z_2\times D_4$.

Table~\ref{table-subgroup_GL32} indicates, and the above discussion confirms, 
that we need to consider only one example of the $\Z_3$ subgroup.
A suitable order-3 automorphism is 
\begin{equation}
{\mathfrak c} = 
\begin{pmatrix}
	0 & 1 & 0\\
	0 & 0 & 1\\
	1 & 0 & 0
\end{pmatrix}
\end{equation}
which, in a suitable basis, leads to 
\begin{equation}
c = 
\begin{pmatrix}
	0 & 1 & 0 & 0\\
	0 & 0 & 1 & 0\\
	1 & 0 & 0 & 0\\
	0 & 0 & 0 & 1\\
\end{pmatrix}\,.\label{E8-c-1}
\end{equation}
This $c$ leads to the split extension $(\Z_2)^3\rtimes\Z_3\simeq A_4 \times \Z_2$, which reduces the potential
$V_1$ to 
\begin{equation}
V_1 = \lambda\left[\fdf{1}{2}^2 + \fdf{2}{3}^2 + \fdf{3}{1}^2\right]
+ \lambda'\left[\fdf{1}{4}^2 + \fdf{2}{4}^2 + \fdf{3}{4}^2\right] + h.c.
\label{E8-V1-Z3}
\end{equation}
Similarly, extension of $\Z_4$ can be constructed with the aid of 
\begin{equation}
{\mathfrak d} = 
\begin{pmatrix}
	0 & 0 & 1\\
	1 & 1 & 1\\
	0 & 1 & 0
\end{pmatrix}\,,
\quad\mbox{leading to} \quad
d = 
\begin{pmatrix}
	0 & 0 & 1 & 0\\
	0 & 0 & 0 & 1\\
	0 & 1 & 0 & 0\\
	1 & 0 & 0 & 0\\
\end{pmatrix}\,,\label{E8-d-1}
\end{equation}
that is, the same cyclic permutation $\phi_1 \mapsto \phi_3 \mapsto \phi_2 \mapsto \phi_4 \mapsto \phi_1$ as we encountered 
in Eq.~\eqref{extension-Z4Z2-2-eq2} leading to the same group $(\Z_2)^3\rtimes\Z_4 \simeq {\tt SmallGroup(32, 6)}$.
The potential $V_1$ simplifies then to 
\begin{equation}
V_1 = \lambda\left[\fdf{1}{2}^2 + \fdf{3}{4}^2\right]
+ \lambda'\left[\fdf{1}{3}^2 + \fdf{3}{2}^2 + \fdf{2}{4}^2 + \fdf{4}{1}^2\right] + h.c.\,,
\label{E8-V1-Z4}
\end{equation}
with real $\lambda$ and, in general, complex $\lambda'$,
while the rephasing-insensitive part $V_0$ takes the form as in Eq.~\eqref{extension-Z4Z2-2-eq3}.
The total potential $V_0 + V_1$ has nine free parameters and leads to a very constrained scalar sector.

Continuing with the subgroups in Table~\ref{table-subgroup_GL32}, we deal next with the group $\Z_2\times\Z_2$, 
which corresponds to two distinct conjugacy classes in $GL(3,2)$.
With the aid of {\ttfamily GAP}, we select these two pairs of generators: 
\begin{equation}
\Z_2\times\Z_2 \subset GL(3,2),\,
\mbox{option 1:}\qquad 
\mathfrak{b}_1 = 
\begin{pmatrix}
	0 & 1 & 0\\
	1 & 0 & 0\\
	0 & 0 & 1
\end{pmatrix}\;,\quad 
\mathfrak{b}_3' 
\,,\label{E8-Z2Z2-1}
\end{equation}
\begin{equation}
\Z_2\times\Z_2 \subset GL(3,2),\,
\mbox{option 2:}\qquad 
\mathfrak{b}_1'\,, \quad 
\mathfrak{b}_3'\,.\label{E8-Z2Z2-2}
\end{equation}
As usual, in each case, we have options for split vs. non-split extensions, similar to the cases we have found earlier.
Skipping technical details, we only provide the final result:
the first option leads only to $(\Z_2)^3\rtimes(\Z_2\times\Z_2) \simeq {\tt SmallGroup(32,49)}$,
constructed either as split or non-split extension,
while the second option produces the groups {\tt SmallGroup(32,27)} (split extension) and {\tt SmallGroup(32,34)} (non-split extension).
The constraints on the potential can also be established using the methods we have already used before.

Moving on, we select a representative $S_3 \subset GL(3,2)$ generated by $b$ in Eq.~\eqref{E8-b-example2}
and $c$ in Eq.~\eqref{E8-c-1}. We checked that only split extension is possible, leading to total symmetry group 
$(\Z_2)^3\rtimes S_3\simeq S_4\times \Z_2$. In essence, this is the same symmetry group $S_4$ acting on the first three doublets,
which we had already in the 3HDM, times the $\Z_2$ group of independent sign flip of $\phi_4$.
The potential $V_1$ in this case is the same as for the $(\Z_2)^3\rtimes \Z_3\simeq A_4\times \Z_2$
and was given in Eq.~\eqref{E8-V1-Z3}. The only extra condition now is that $\lambda$ in Eq.~\eqref{E8-V1-Z3} must be real.

A representative subgroup $D_4 \subset GL(3,2)$ can be generated by the same $\mathfrak{d}$ as in Eq.~\eqref{E8-d-1} and $\mathfrak{b}_1$
in Eq.~\eqref{E8-Z2Z2-1} as they satisfy $\mathfrak{b}_1^{-1}\mathfrak{d}\mathfrak{b}_1 = \mathfrak{d}^{-1}$.
The resulting group is $G = (\Z_2)^3 \rtimes D_4 \simeq UT(4,2) \simeq {\tt SmallGroup(64, 138)}$, the same group
as in Eq.~\eqref{extension-Z4Z2-2-D4}. The potential $V_0$ takes the form \eqref{extension-Z4Z2-2-eq3} while the $V_1$ part 
is the same as in Eq.~\eqref{E8-V1-Z4} but now both $\lambda$ and $\lambda'$ being real.

The next subgroup to consider is $A_4 \subset GL(3,2)$. Table~\ref{table-subgroup_GL32} indicates two distinct conjugacy classes
for $A_4$. However it turns out that the same constraint which forbade the $\Z_7$ subgroup
forbids also one of the $A_4$ conjugacy classes.
The remaining one can be generated by the familiar $\mathfrak{b}_3'$ in Eq.~\eqref{E8-b'-examples} and $\mathfrak{c}$ in Eq.~\eqref{E8-c-1}
because they satisfy the relations defining the $A_4$ group:
$(\mathfrak{b}_3')^{3} = \mathfrak{c}^2 = (\mathfrak{cb'}_3)^3 = \mathfrak{e}$.
In terms of $b$ and $c$, we arrive at the defining presentation of $A_4$ as a group
of even permutations of four doublets.

The total symmetry group obtained through this construction is $(\Z_3)^3\rtimes A_4 \simeq {\tt SmallGroup(92,70)}$ of order 92.
Imposing invariance under the $A_4$ group of permutations and sign flips of individual doublets 
dramatically constrains the potential, with 
\begin{eqnarray}
V_0 &=& m^2 \left(\fdfn{1}{1} + \fdfn{2}{2} + \fdfn{3}{3} + \fdfn{4}{4}\right)
+ \Lambda \left(\fdfn{1}{1} + \fdfn{2}{2} + \fdfn{3}{3} + \fdfn{4}{4}\right)^2\nonumber\\
&&+\Lambda' \left[(\fdfn{1}{1})^2 + (\fdfn{2}{2})^2 + (\fdfn{3}{3})^2 + (\fdfn{4}{4})^2\right]\nonumber\\
&& + \Lambda'' \left(|\fdfn{1}{3}|^2 + |\fdfn{2}{3}|^2 + |\fdfn{1}{4}|^2 + |\fdfn{2}{4}|^2 
+ |\fdfn{1}{2}|^2 + |\fdfn{3}{4}|^2\right)\label{E8-A4-V0}
\end{eqnarray}
and 
\begin{equation}
V_1 = \lambda\left[\fdf{1}{2}^2 + \fdf{2}{3}^2 + \fdf{3}{1}^2 + \fdf{1}{4}^2 + \fdf{2}{4}^2 + \fdf{3}{4}^2 + h.c.\right]\,.
\label{E8-A4-V1}
\end{equation}
However, upon a quick inspection, it becomes clear that this potential is invariant under {\em all} permutations of the four doublets,
not only the positive-signature ones. Therefore, the group $(\Z_3)^3\rtimes A_4$ we just constructed is not realizable in the 4HDM because
it automatically leads a additional discrete symmetries. The total symmetry content of this potential is
\begin{equation}
G = (\Z_2)^3\rtimes S_4 \simeq {\tt SmallGroup(192,955)}\,.\label{E8-S4}
\end{equation}
It is remarkable that the 4HDM scalar sector, with so many symmetries and so few free parameters, 
does not possess an accidental continuous symmetry.

\section{Discussion and conclusions}\label{section-conclusions}

Global symmetries are a powerful feature of multi-Higgs-doublet models as they often lead to symmetry-induced phenomenological features.
The symmetry options available in the 2HDM and 3HDM have been explored in hundreds of papers, 
and the methods developed there can be also used in even more elaborate scalar sectors.
There is also a significant literature on models with four Higgs doublets,
almost all of them based on a specific symmetry group
(see a historical overview in our previous paper on the subject \cite{Shao:2023oxt}).
However no attempt has been made up to now to classify symmetry options available in the 4HDM.
We started in \cite{Shao:2023oxt} --- and continue in the present paper ---
our quest for classification of finite non-abelian symmetry groups
which can be imposed on the scalar sector of the 4HDM without causing accidental continuous symmetries.

Here, we employ the same group extension technique which was so successful in the 3HDM \cite{Ivanov:2012ry,Ivanov:2012fp}.
Namely, we start with the finite abelian symmetry groups $A$ realizable in the 4HDM, which are known from \cite{Ivanov:2011ae}
and given in Table~\ref{table-abelian}, 
and then constructed non-abelian groups as group extensions of the form
$A\rtimes K$ (split extensions) or $A\,.\,K$ (non-split extensions), where $K \subseteq \Aut(A)$.
In the previous paper \cite{Shao:2023oxt}, we found all finite non-abelian groups
which emerge from cyclic groups $A$ (see the left half of Table~\ref{table-abelian}).
In the present work, we extended this analysis to the rephasing groups $A$ 
which are products of cyclic groups, namely,
to $A = \Z_2\times\Z_2$, $\Z_4\times\Z_2$, and $(\Z_2)^3$.

\begin{table}[H]
\centering
\begin{tabular}[t]{cccc}
	\toprule
	$A$ & extensions & $G$ & $|G|$ \\
	\midrule
	\multirow{3}{*}{$\Z_2\times\Z_2$} & $A\rtimes\Z_2$ & $D_4$ & 8 \\[1mm]
	& $A\rtimes\Z_3$ & $A_4$ & 12  \\[1mm]
	& $A\rtimes S_3$ & $S_4$ & 24 \\[1mm]
	\midrule
	\multirow{5}{*}{$\Z_4\times\Z_2$} 
	& $A\rtimes\Z_2$ & $\Z_2\times D_4$ & 16 \\[1mm]
	& $A\rtimes\Z_2$ & {\tt SmallGroup(16,3)} & 16 \\[1mm]
	& $A\rtimes\Z_2$ & {\tt SmallGroup(16,13)} & 16 \\[1mm]
	& $A\rtimes\Z_4$ & {\tt SmallGroup(32,6)} & 32 \\[1mm]
	& $A\rtimes (\Z_2\times\Z_2)$ & {\tt SmallGroup(32,49)} & 32 \\[1mm]
	& $A\rtimes (\Z_2\times\Z_2)$ & {\tt SmallGroup(32,27)} & 32 \\[1mm]
	& $A\rtimes D_4$ & $UT(4,2)$ & 64 \\[1mm]
	\midrule
	\multirow{6}{*}{$(\Z_2)^3$} & $A\rtimes\Z_2$ & $\Z_2\times D_4$ & 16 \\[1mm]
	& $A\,.\,\Z_2$ & {\tt SmallGroup(16,3)} & 16 \\[1mm]
	& $A\rtimes\Z_3$ & $\Z_2\times A_4$ & 24 \\[1mm]
	& $A\rtimes\Z_4$ &{\tt SmallGroup(32,6)} & 32 \\[1mm]
	& $A\rtimes(\Z_2\times\Z_2)$ & {\tt SmallGroup(32,49)} & 32 \\[1mm]
	& $A\rtimes(\Z_2\times\Z_2)$ & {\tt SmallGroup(32,27)} & 32 \\[1mm]
	& $A\,.\,(\Z_2\times\Z_2)$ & {\tt SmallGroup(32,34)} & 32 \\[1mm]
	& $A\rtimes S_3$ & $\Z_2\times S_4$ & 48 \\[1mm]
	& $A\rtimes D_4$ & $UT(4,2)$ & 64 \\[1mm]
	& $A\rtimes S_4$ & {\tt SmallGroup(192,955)} & 192 \\[1mm]
	\bottomrule
\end{tabular}
\caption{The summary table of finite non-abelian groups $G$ in the 4HDM scalar sector
	constructed as extension by $\Z_2\times\Z_2$, $\Z_4\times\Z_2$, and $\Z_2\times\Z_2\times\Z_2$. 
	Non-split extensions of the form $G = A\,.\,K$ are
	shown only when they lead to groups $G$ that cannot be obtained by a split extension with the same $A$.}
\label{table-non-abelian}
\end{table}

Table~\ref{table-non-abelian} summarizes the results of this paper.
As can be seen, there are many extensions fitting inside $PSU(4)$. 
Some of the resulting groups have their own labels, but many of these groups are usually
referred to using their {\tt GAP} id. The main text describes how each symmetry group arises, what its generators are, 
and how the potential invariant under this group can be constructed. In certain cases, we gave the potentials
explicitly; in other cases their form can be directly reconstructed using the relations described in the paper.
Should the reader be interested, a phenomenological study can be conducted for any of these symmetry-shaped 4HDMs.
We welcome the community to explore phenomenological and cosmological features of large symmetry groups found here.

Not only did we classify the finite non-abelian symmetry groups which emerge in the 4HDM scalar sector as extensions
by all rephasing groups,
but we also further developed the methods which are not often used in BSM model building.
We believe that these techniques represent a useful contribution to symmetry-based model building
and can be exploited in other settings.

It must be stressed that this paper, together with \cite{Shao:2023oxt}, does not yet complete the classification 
of all realizable finite non-abelian symmetry groups in the 4HDM scalar sector.
These two papers only deal with non-abelian groups which can be constructed as extensions based on rephasing groups $A$.
However Table~\ref{table-abelian} contains three more abelian subgroups of $PSU(4)$ 
which cannot be represented by rephasing transformations alone.
Moreover, we do not yet have the proof that these three additional groups exhaust
all abelian subgroups of $PSU(4)$ whose full pre-images in $SU(4)$ are non-abelian.
This proof and the explicit constructions are delegated to a follow-up paper and
may require development of yet another set of techniques.

Finally, as explained in Section~\ref{subsection-extension}, the group extension technique does not exhaust 
all finite non-abelian groups for the 4HDM. In the 3HDM, this technique provided an exhaustive classification
due to Burnside’s $p^aq^b$-theorem. In the 4HDM, this theorem no longer applies. As a result,
it may happen that other finite groups could be found, which are not of the form $A\,.\,K$, with $K \subseteq \Aut(A)$.
Systematic exploration of such cases must rely on a different strategy.

We hope, however, that these additional groups will be relatively few. 
In this sense, the results of \cite{Shao:2023oxt} and the present work most likely cover the vast majority
of finite non-abelian symmetry groups upon which 4HDM scalar sectors could be constructed. 

\section*{Acknowledgments}

This work was supported by the Guangdong Natural Science Foundation (project No. 2024A1515012789) and 
by the Fundamental Research Funds for the Central Universities, Sun Yat-sen University, China.
M.K. acknowledges a visit to the School of Physics and Astronomy during which a part of this work was done.
We also would like to thank the referees for their valuable comments and suggestions.


\begin{thebibliography}{99}

\bibitem{Lee:1973iz}
T.~D.~Lee,
Phys. Rev. D \textbf{8}, 1226-1239 (1973)
doi:10.1103/PhysRevD.8.1226

\bibitem{Branco:2011iw}
G.~C.~Branco, P.~M.~Ferreira, L.~Lavoura, M.~N.~Rebelo, M.~Sher and J.~P.~Silva,
Phys. Rept. \textbf{516}, 1-102 (2012)
doi:10.1016/j.physrep.2012.02.002
[arXiv:1106.0034 [hep-ph]].

\bibitem{Weinberg:1976hu}
S.~Weinberg,
Phys. Rev. Lett. \textbf{37}, 657 (1976)
doi:10.1103/PhysRevLett.37.657

\bibitem{Ivanov:2017dad}
I.~P.~Ivanov,
Prog. Part. Nucl. Phys. \textbf{95}, 160-208 (2017)
doi:10.1016/j.ppnp.2017.03.001
[arXiv:1702.03776 [hep-ph]].

\bibitem{Bjorken:1977vt}
J.~D.~Bjorken and S.~Weinberg,
Phys. Rev. Lett. \textbf{38}, 622 (1977)
doi:10.1103/PhysRevLett.38.622

\bibitem{Shao:2023oxt}
J.~Shao and I.~P.~Ivanov,
JHEP \textbf{10}, 070 (2023)
doi:10.1007/JHEP10(2023)070
[arXiv:2305.05207 [hep-ph]].

\bibitem{Derman:1979nf}
E.~Derman and H.~S.~Tsao,
Phys. Rev. D \textbf{20}, 1207 (1979)
doi:10.1103/PhysRevD.20.1207

\bibitem{Segre:1978ji}
G.~Segre and H.~A.~Weldon,
Phys. Lett. B \textbf{83}, 351-354 (1979)
doi:10.1016/0370-2693(79)91125-0

\bibitem{GonzalezFelipe:2013xok}
R.~Gonz\'alez Felipe, H.~Ser\^odio and J.~P.~Silva,
Phys. Rev. D \textbf{87}, no.5, 055010 (2013)
doi:10.1103/PhysRevD.87.055010
[arXiv:1302.0861 [hep-ph]].

\bibitem{Ivanov:2014doa}
I.~P.~Ivanov and C.~C.~Nishi,
JHEP \textbf{01}, 021 (2015)
doi:10.1007/JHEP01(2015)021
[arXiv:1410.6139 [hep-ph]].

\bibitem{Leurer:1992wg}
M.~Leurer, Y.~Nir and N.~Seiberg,
Nucl. Phys. B \textbf{398}, 319-342 (1993)
doi:10.1016/0550-3213(93)90112-3
[arXiv:hep-ph/9212278 [hep-ph]].

\bibitem{GonzalezFelipe:2014mcf}
R.~Gonz\'alez Felipe, I.~P.~Ivanov, C.~C.~Nishi, H.~Ser\^odio and J.~P.~Silva,
Eur. Phys. J. C \textbf{74}, no.7, 2953 (2014)
doi:10.1140/epjc/s10052-014-2953-9
[arXiv:1401.5807 [hep-ph]].

\bibitem{Ivanov:2015mwl}
I.~P.~Ivanov and J.~P.~Silva,
Phys. Rev. D \textbf{93}, no.9, 095014 (2016)
doi:10.1103/PhysRevD.93.095014
[arXiv:1512.09276 [hep-ph]].

\bibitem{Haber:2018iwr}
H.~E.~Haber, O.~M.~Ogreid, P.~Osland and M.~N.~Rebelo,
JHEP \textbf{01}, 042 (2019)
doi:10.1007/JHEP01(2019)042
[arXiv:1808.08629 [hep-ph]].

\bibitem{Branco:1999fs}
G.~C.~Branco, L.~Lavoura and J.~P.~Silva,
Int. Ser. Monogr. Phys. \textbf{103}, 1-536 (1999)

\bibitem{Branco:1983tn}
G.~C.~Branco, J.~M.~Gerard and W.~Grimus,
Phys. Lett. B \textbf{136}, 383-386 (1984)
doi:10.1016/0370-2693(84)92024-0

\bibitem{deMedeirosVarzielas:2011zw}
I.~de Medeiros Varzielas and D.~Emmanuel-Costa,
Phys. Rev. D \textbf{84}, 117901 (2011)
doi:10.1103/PhysRevD.84.117901
[arXiv:1106.5477 [hep-ph]].

\bibitem{deMedeirosVarzielas:2012rxf}
I.~de Medeiros Varzielas,
JHEP \textbf{08}, 055 (2012)
doi:10.1007/JHEP08(2012)055
[arXiv:1205.3780 [hep-ph]].

\bibitem{Ivanov:2013nla}
I.~P.~Ivanov and L.~Lavoura,
Eur. Phys. J. C \textbf{73}, no.4, 2416 (2013)
doi:10.1140/epjc/s10052-013-2416-8
[arXiv:1302.3656 [hep-ph]].

\bibitem{Das:2014fea}
D.~Das and U.~K.~Dey,
Phys. Rev. D \textbf{89}, no.9, 095025 (2014)
[erratum: Phys. Rev. D \textbf{91}, no.3, 039905 (2015)]
doi:10.1103/PhysRevD.89.095025
[arXiv:1404.2491 [hep-ph]].

\bibitem{Darvishi:2021txa}
N.~Darvishi, M.~R.~Masouminia and A.~Pilaftsis,
Phys. Rev. D \textbf{104}, no.11, 115017 (2021)
doi:10.1103/PhysRevD.104.115017
[arXiv:2106.03159 [hep-ph]].

\bibitem{Das:2021oik}
D.~Das, P.~M.~Ferreira, A.~P.~Morais, I.~Padilla-Gay, R.~Pasechnik and J.~P.~Rodrigues,
JHEP \textbf{11}, 079 (2021)
doi:10.1007/JHEP11(2021)079
[arXiv:2106.06425 [hep-ph]].

\bibitem{Ivanov:2011ae}
I.~P.~Ivanov, V.~Keus and E.~Vdovin,
J. Phys. A \textbf{45}, 215201 (2012)
doi:10.1088/1751-8113/45/21/215201
[arXiv:1112.1660 [math-ph]].

\bibitem{Ivanov:2012fp}
I.~P.~Ivanov and E.~Vdovin,
Eur. Phys. J. C \textbf{73}, no.2, 2309 (2013)
doi:10.1140/epjc/s10052-013-2309-x
[arXiv:1210.6553 [hep-ph]].

\bibitem{Ivanov:2012ry}
I.~P.~Ivanov and E.~Vdovin,
Phys. Rev. D \textbf{86}, 095030 (2012)
doi:10.1103/PhysRevD.86.095030
[arXiv:1206.7108 [hep-ph]].

\bibitem{Darvishi:2019dbh}
N.~Darvishi and A.~Pilaftsis,
Phys. Rev. D \textbf{101}, no.9, 095008 (2020)
doi:10.1103/PhysRevD.101.095008
[arXiv:1912.00887 [hep-ph]].

\bibitem{Branco:2010tx}
G.~C.~Branco, D.~Emmanuel-Costa and C.~Simoes,
Phys. Lett. B \textbf{690}, 62-67 (2010)
doi:10.1016/j.physletb.2010.05.009
[arXiv:1001.5065 [hep-ph]].

\bibitem{Ivanov:2013bka}
I.~P.~Ivanov and C.~C.~Nishi,
JHEP \textbf{11}, 069 (2013)
doi:10.1007/JHEP11(2013)069
[arXiv:1309.3682 [hep-ph]].

\bibitem{Ferreira:2010hy}
P.~M.~Ferreira, M.~Maniatis, O.~Nachtmann and J.~P.~Silva,
JHEP \textbf{08}, 125 (2010)
doi:10.1007/JHEP08(2010)125
[arXiv:1004.3207 [hep-ph]].

\bibitem{Ferreira:2010yh}
P.~M.~Ferreira, H.~E.~Haber, M.~Maniatis, O.~Nachtmann and J.~P.~Silva,
Int. J. Mod. Phys. A \textbf{26}, 769-808 (2011)
doi:10.1142/S0217751X11051494
[arXiv:1010.0935 [hep-ph]].

\bibitem{Ivanov:2018qni}
I.~P.~Ivanov and M.~Laletin,
Phys. Rev. D \textbf{98}, no.1, 015021 (2018)
doi:10.1103/PhysRevD.98.015021
[arXiv:1804.03083 [hep-ph]].

\bibitem{deMedeirosVarzielas:2021zqs}
I.~de Medeiros Varzielas, I.~P.~Ivanov and M.~Levy,
Eur. Phys. J. C \textbf{81}, no.10, 918 (2021)
doi:10.1140/epjc/s10052-021-09681-w
[arXiv:2107.08227 [hep-ph]].

\bibitem{Plantey:2024gju}
R.~Plantey and M.~A.~Solberg,
[arXiv:2407.05085 [hep-ph]].

\bibitem{Rotman}
J.~J.~Rotman,
An Introduction to the Theory of Groups,
Graduate Texts in Mathematics \textbf{148}, 4th Edition, Springer-Verlag, New York, 1995.

\bibitem{GAP}
The GAP Group.
\textit{GAP---Groups, Algorithms, Programming---A System
	for Computational Discrete Algebra. Version 4.11.1; 2022}.
Available at 
{\tt https://www.gap-system.org}.

\bibitem{TheCode}
Jiazhen Shao, The 4HDM Toolbox,
available at {\tt  https://github.com/JiazhenShao/4HDM-Toolbox.git}


\end{thebibliography}
\end{document}